%% file: main.tex
\documentclass[hidelinks,onefignum,onetabnum]{siamart251216}
\usepackage{xcolor}

\input{ex_shared}

\ifpdf
\hypersetup{
  pdftitle={Control Systems with Model Hierarchies: A
Category-Theoretic Framework},
  pdfauthor={Baike She, Tyler Hanks, James
Fairbanks, and Matthew Hale}
}
\fi

\begin{document}

\maketitle

\begin{abstract}
Storage functions are widely used to relate a control system to a simpler model for analysis and
controller design, and to quantify model mismatch when a controller is designed on one
model but implemented on another. In applications, several models of the same physical process
are available at different levels of refinement, while refinement relations, consisting of a storage function and an associated interface, are constructed only between neighboring levels. Thus, a direct refinement relation may be unavailable between the model used for controller design and the model used for implementation.
To address this gap, we consider open control systems with internal ports for interconnection and external ports for input-output analysis. We define a category, $\mathsf{OCS}_{\mathrm{chain}}$, whose objects are open control systems and whose morphisms are finite chains of refinement relations. Using this category, we derive an output mismatch bound and a closed-loop performance result for controllers refined through a sequence of interfaces. We then define a functor that transfers refinement relations from open control systems to systems without internal ports. We also develop a network lifting result showing that, under standard interconnection compatibility conditions, subsystem storage functions can be combined to obtain refinement relations and performance bounds for composite networks. Two thermal network examples illustrate recursive controller refinement and the propagation of model mismatch along refinement chains. Simulations show that an HVAC system tracks a time-varying reference temperature with error below $0.2^\circ\mathrm{C}$ even when the controller is designed on the coarsest model and deployed on the finest model.
\end{abstract}



\section{Introduction}
A physical process can often be modeled at multiple levels of detail, ranging from simplified reduced-order descriptions to high-fidelity models. For example, work in~\cite{ahmed2021multifidelity} models transport processes governed by the same underlying physics using a full-order model together with a reduced-order model. Similarly, the same infectious disease dynamics can be modeled at different scales using either a group compartmental model or network compartmental models~\cite{chen2025scalable}. To influence the behavior of physical processes, controllers are typically designed using mathematical models. This design is often carried out on simplified models, while implementation is performed on more refined models or on the original system. Consequently, the mismatch between the simplified models used for design and the system used for implementation can affect stability and performance~\cite{kuntz2025beyond, Fan62265, she2022optimal}.

One way to address this gap is to study refinement relations between the original system and models at different levels of detail so that a controller designed on a model at one level of detail  can also be implemented on the original 
system~\cite{girard2009hierarchical,tabuada2008approximate}. 
In many applications, multiple models of the same control system are available and are naturally ordered by fidelity or complexity. For example, the same process can be represented by a reduced-order model, one or more intermediate-fidelity models, and a full-detail model~\cite{coulier2022systematic, chen2025scalable}. In these settings, such refinement relations are often available only between models at
neighboring levels of refinement in the hierarchy, e.g., between two consecutive levels that differ by a single modeling step, such as one state aggregation, one reduction in model order, or one coarsening of a network representation~\cite{firouzmand2024hierarchical, niu2025bridging}. Establishing such a relation directly between a coarse model and a much more refined one is often difficult~\cite{lavaei2019compositional}. Consequently, when controller design and implementation are separated by several levels in a model hierarchy, it is difficult to establish analytical guarantees that relate the behavior of the implementation model to that of the design model.
This difficulty becomes even more pronounced in networked systems, where models are formed by interconnecting subsystems~\cite{arcak2016networks, cervera2007interconnection,she2025dissipativity}. Even when refinement relations between corresponding subsystems are available, constructing a refinement relation directly between the corresponding interconnected models can still be challenging. 

Motivated by these challenges, we develop a category-theoretic framework for controller refinement over model hierarchies that emphasizes how refinement relations between models can be composed, rather than treating the models in isolation. Category theory focuses on relationships between systems rather than on the systems alone, and in dynamical systems and control it has been used to model abstractions, compositions, and interconnections within a common mathematical framework~\cite{ames2006categorical,wisniewski2006towards,culbertson2019formal,bakirtzis2021categorical,hanks2024modeling}.
Category theory uses morphisms between objects to describe relationships between them.
In this paper, each elementary morphism represents a pairwise refinement
relation between two control systems and is specified by a storage
function together with a certain type of interface between those systems
(see Section~\ref{Sec_Similar_Control_Systems} for details). 
General morphisms are finite chains
of such elementary morphisms. The key idea is to treat these morphisms, i.e., refinement relations between models of the same control system at different levels of detail, as the primary carriers of information, since they encode how 
state and output trajectories, control inputs, and output bounds are transferred across a model hierarchy. 
Under this viewpoint, a model hierarchy 
is represented by a collection of composable morphisms,
and a multi-step refinement 
between non-adjacent models is given by a composition of such morphisms. This formulation makes it possible to study controller refinement through the composition of morphisms, even when no direct refinement relation is available between models separated by several levels in the hierarchy.

The main contributions of this paper are as follows:
\begin{itemize}
\item We introduce a category-theoretic framework for refinement over model hierarchies of open
control systems with both internal and external ports. In particular, we define the category
$\mathsf{OCS}_{\mathrm{chain}}$, whose objects are open control systems and whose morphisms are chains
of refinement relations given by storage functions and their associated interfaces. We prove
that this construction satisfies all axioms neede to be a category 
(Theorem~\ref{Thm_CTRL_Path}).

\item We develop controller refinement results for controllers designed on simpler models and
implemented on more refined models. We show that the output mismatch between
the two endpoint models of a chain is bounded explicitly in terms of the storage
functions along the chain. We further show how this endpoint mismatch yields a quantitative bound on the
transfer of closed-loop performance from one model to another under refinement 
(Theorem~\ref{Thm_Control_Refine}
Corollary~\ref{Cor_CL_Performance_Transfer}).

\item We formulate an analogous chain category called $\mathsf{CS}_{\mathrm{chain}}$ for control systems without internal ports, and 
we define
a functorial projection from $\mathsf{OCS}_{\mathrm{chain}}$ to $\mathsf{CS}_{\mathrm{chain}}$. This projection
removes the structure of the internal ports while preserving identity morphisms, composition, and the
structure of refinement chains. Consequently, the controller refinement and performance-transfer
results remain applicable to the corresponding projected systems without internal ports
(Theorem~\ref{Def:CTRL-O-path},
Theorem~\ref{Thm:PhiFunctor},
Corollary~\ref{Cor:PhiTransfer_1}, and
Corollary~\ref{Cor:PhiExplicit_2}).

\item We connect subsystem-level refinement relations with refinement 
relations
between interconnected network
models. Using certain compositional conditions from
\cite{zamani2017compositional}, we show that weighted sums of storage functions
between corresponding open control subsystems 
can be used to 
define simulation functions between the
corresponding composite control systems. We then express these network-level simulation
functions as elementary simulation morphisms in
$\mathsf{CS}_{\mathrm{chain}}$ and show that they can be composed across a hierarchy of network models, which
provides a scalable route for controller refinement across multiple levels in a network model hierarchy 
(Proposition~\ref{Prop:Composite_Control} and
Corollary~\ref{Cor:NetworkHierarchyRefinement}).

\item We illustrate the aforementioned compositional framework using two thermal network examples. The first considers a
hierarchy of models without internal ports, while the second considers interconnected room
models and combines storage functions for the subsystems at the network level. 
We show in both examples 
that an HVAC system tracks a time-varying reference temperature with error less than~$0.2^\circ\mathrm{C}$ 
when using controllers that are refined using the framework we develop
(Section~\ref{Sec_Example_Simulation}). 
\end{itemize}

\subsubsection*{Related Work}
In control system synthesis, controller design is often carried out on simplified models, called abstractions, while implementation is performed on more refined models or on the original system.
For example, 
the work in~\cite{pola2008approximately} shows that incrementally stable nonlinear systems admit approximately bisimilar symbolic 
abstractions for controller synthesis, which can simplify control design. 
Similarly, work in~\cite{girard2009hierarchical} studies a complex system together with a simpler abstraction and shows how a simulation function and an interface can be used to design a controller on the abstraction and then lift it to the original system. Moreover, work in~\cite{hashimoto2022learning} constructs  abstractions by learning unmodeled dynamics from data while preserving the behavioral relation needed for safe control synthesis. Likewise, developments in~\cite{kazemi2024data} study formal controller synthesis by using data to build an abstraction and then synthesize controllers from that abstraction. When a full-order controller is too expensive for real-time use, such as in flow control problems, results in~\cite{sasaki2024design} construct a reduced-order controller that performs nearly as well as the full-order one. In addition, storage functions provide a standard tool for relating a control system to a simpler model through an interface and an output mismatch bound, thereby supporting analysis and controller refinement of the higher-order system based on the simpler model~\cite{zamani2017compositional}.
\textcolor{black}{Building upon the notation of storage functions in~\cite{zamani2017compositional}, we extend pairwise measure of  model mismatch to chains of systems using category theory.}
All of the aforementioned works study relations among models in some form, though to the best of our knowledge
the current paper is the first to use category theory to model structures of systems and refinement
relations. The formalisms we develop allow for the systematic construction
of new refinement relations,~\textcolor{black}{organize chains of models and refinement relations into a unified mathematical structure, as well as enable refinement relations to be composed across multiple levels of a model hierarchy.}


The rest of the paper is organized as follows. Section~\ref{sec:background} presents background on control systems, storage functions, and category theory, and formulates the problems addressed in this paper. Section~\ref{Sec_CTRL} presents the main results on the category $\mathsf{OCS}_{\mathrm{chain}}$ and the associated controller refinement results along chains. Section~\ref{Sec_CTRLO} introduces the 
category~$\mathsf{CS}_{\mathrm{chain}}$ and the 
functorial projection 
from $\mathsf{OCS}_{\mathrm{chain}}$ to $\mathsf{CS}_{\mathrm{chain}}$, and then it 
develops the corresponding framework for control systems without internal ports. Section~\ref{Sec:NetworkLifting} presents network lifting results, showing how subsystem-level chain morphisms can be carried to interconnected network models. Section~\ref{Sec_Example_Simulation} provides simulation examples to illustrate the main results. Finally, Section~\ref{Sec_Conclusion} concludes the paper.

\section{Preliminaries and Problem Formulation}
\label{sec:background}
In this section, we introduce background on open control systems, refinement relations between models, and category theory. We then formulate the problems we solve in this work. We begin by presenting the notation that will be used throughout the paper.

\subsection{Notation}
We use $\mathbb{R}$ to denote the set of real numbers, $\mathbb{R}_{\geq 0}$ to denote the nonnegative reals, and $\mathbb{R}_{> 0}$ to denote the positive reals. We use $\mathbb{N}$ to denote the natural numbers including $0$, and $\mathbb{N}_{> 0}$ to denote the positive natural numbers. All vectors are column vectors.
We define the index set $\underline{N} := \{1,2,3,\dots,N\}$ for any $N \in \mathbb{N}_{>0}$.
Given a vector $x \in \mathbb{R}^n$, we use $(x)_i$ to denote its $i^{th}$ entry.
In addition, we use $\|x\|$ to denote the Euclidean norm of $x \in \mathbb{R}^n$.
We use $I_n \in \mathbb{R}^{n \times n}$ to denote the $n \times n$ identity matrix. For a matrix $A \in \mathbb{R}^{n \times n}$, we use $A_{i,j}$ or $(A)_{i,j}$ to denote its entry in the $i^{th}$ row and $j^{th}$ column. If a matrix $A \in \mathbb{R}^{n\times n}$ is partitioned into four blocks as
\[
A=
\begin{bmatrix}
[A]_{11} & [A]_{12}\\
[A]_{21} & [A]_{22}
\end{bmatrix},
\]
then $[A]_{i,j}$ with $i,j\in\{1,2\}$ denotes the block in the $(i,j)^{th}$ position.
If a partitioned matrix contains more blocks, then they are indexed analogously. 
We use $A \prec 0$ and $A \preceq 0$ to denote that $A$ is negative definite and negative semidefinite, respectively.
Given $N \in \mathbb{N}_{>0}$, along with  $n^o \in \mathbb{N}_{>0}$
and $x^o \in \mathbb{R}^{n^o}$ for each~$o \in \underline{N}$, 
we use $\mathbf{x} = [(x^1)^\top, \dots, (x^N)^\top]^\top$ to denote the concatenated vector $\mathbf{x} \in \mathbb{R}^{\mathbf{n}}$, where $\mathbf{n} = \sum_{o=1}^N n^o$.
A continuous function $\alpha:\mathbb R_{\ge0}\to\mathbb R_{\ge0}$ belongs to class 
$\mathcal K_\infty$ if it is strictly increasing, satisfies $\alpha(0)=0$, and
$\alpha(r)\to\infty$ as $r\to\infty$. We define the zero function $f_0(\cdot)$ such that $f_0(x)=0$ for all $x\in\mathbb{R}^n$.

\subsection{Open Control Systems and Storage Functions}
\label{Sec_Similar_Control_Systems}
Many systems of interest are composite and can be modeled as open control systems with internal ports that model interconnections among subsystems, together with external signals that can be used for analysis and the computation of control inputs. To study  refinement problems, we use storage functions to characterize refinement relations between pairs of system models. A storage function provides an interface for refining inputs from a coarse model to a more refined one, as well as a bound on the distance between these 
systems' external outputs. The storage
function formalism is 
the basis for the multi-model and compositional results developed in this paper. 

We begin by defining open control systems with both external and internal inputs and outputs. We adopt the following definition of control systems and related simpler models from~\cite{zamani2017compositional}.
\begin{definition}[{Open control system; \cite[Definition 2.1]{zamani2017compositional}}]
\label{Def_Con_Systm}
An open control system $\Delta$ is a tuple 
\begin{align*}
 \Delta = (\mathbb R^n, \mathbb R^m, \mathbb R^p, \mathcal{U}, \mathcal{W}, f, \mathbb R^{q},  \mathbb R^{s}, h, y),   
\end{align*}
where $\mathbb R^n$ is the state space, 
$\mathbb R^m$ is the external input space, 
$\mathbb R^p$ is the internal input space, 
$\mathbb R^{q}$ is the external output space, 
and $\mathbb R^{s}$ is the internal output space, and
\begin{itemize}
        \item $\mathcal{U}$ is a non-empty
        subset of the set of all measurable functions
of time from open intervals in  $\mathbb R$ to $\mathbb R^m$; 
        \item  $\mathcal{W}$ is a subset of the set of all measurable functions
of time from open intervals in  $\mathbb R$ to $\mathbb R^p$ such that~$\mathcal{W}$
contains the constant zero signal; 
        \item  $f: \mathbb R^n \times \mathbb R^m \times \mathbb R^p \rightarrow \mathbb R^n$ is a continuous map satisfying
the following Lipschitz assumption: for every compact set $D\subseteq \mathbb R^n$, there exists a constant $\tau\in \mathbb R_{>0}$ such that $\|f(x^1,u,w)-f(x^2,u,w)\|\leq \tau \|x^1-x^2\|$ for all $x^1, x^2\in D$, all $u\in \mathbb R^m$, and all $w\in \mathbb R^p$;
\item $h: \mathbb R^n\rightarrow \mathbb R^{q}$ is the continuous external output map;
\item $y: \mathbb R^n\rightarrow \mathbb R^{s}$ is the internal continuous  output map.
    \end{itemize}
\hfill $\lozenge$
\end{definition}

We illustrate Definition~\ref{Def_Con_Systm} in the following example. 

\begin{example}
For example, consider two control systems~$\Delta_1$ and $\Delta_2$, as shown in Figure~\ref{fig:example_two_control_systems}.
This figure demonstrates that both systems are equipped with external and internal inputs and outputs. When constructing a composite control system, the internal output of one system is fed into the internal input of the other system. 
Figure~\ref{fig:example_two_control_systems} highlights how internal inputs and outputs 
enable the composition of 
control systems. After composition, the external inputs 
and outputs of~$\Delta_1$ and~$\Delta_2$ 
remain available for control, while the internal
inputs of each system 
are determined by the 
internal outputs to which they have been connected. 

\begin{figure}[!h]
    \centering
    \includegraphics[width=0.65\linewidth]{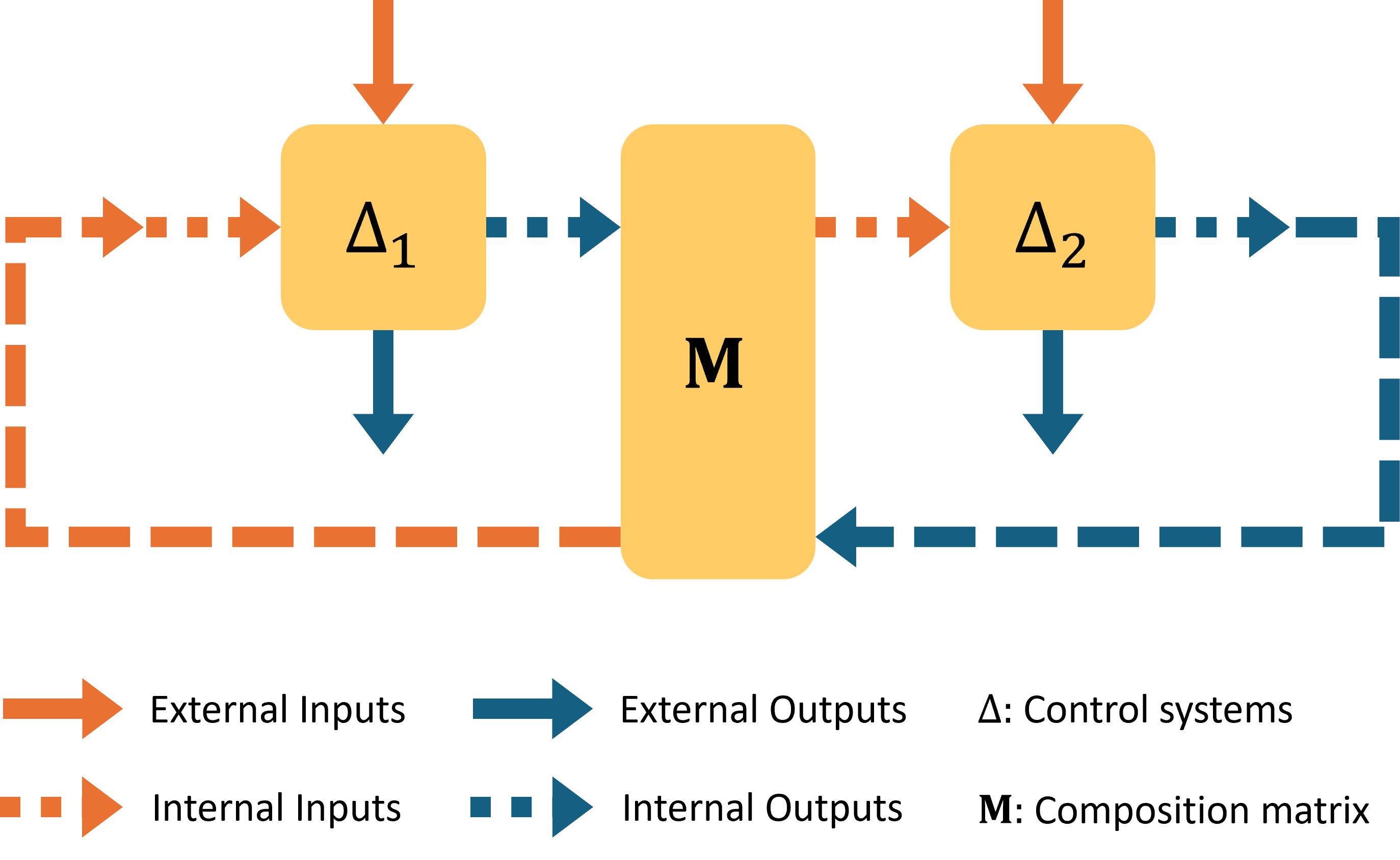}
    \caption{Two control systems and their composite control system. The internal output of $\Delta_1$ is fed into the internal input of $\Delta_2$, and this operation is
    modeled by the composition matrix $\mathbf{M}$, which
    will be discussed in Section~\ref{Sec_Background_Composite_Control_System}. 
    Similarly, the internal output of $\Delta_2$ is fed into the internal input of $\Delta_1$. 
    After composing~$\Delta_1$ and~$\Delta_2$, 
    the external inputs and outputs of both control systems become the inputs and outputs of the composite control system.}
    \label{fig:example_two_control_systems}
\end{figure} 
\end{example}

A locally absolutely continuous curve $x: (a, b)\rightarrow \mathbb R^n$ is a state
trajectory of an open control system $\Delta$ if there exist input trajectories $u(\cdot)\in \mathcal{U}$
and $w(\cdot)\in \mathcal{W}$ such that, for almost all  $t\in (a, b)$,
\begin{align}
\label{Eq: Control}
     \Delta:
    \begin{cases}
\dot x(t)=f\big(x(t),u(t),w(t)\big),\\
h(t)=h\big(x(t)\big),\\
y(t)=y\big(x(t)\big).
    \end{cases} 
\end{align}

 We call $(x,h,y, u, w)$ a trajectory 
 of  $\Delta$, consisting of a state trajectory $x$, output trajectories $h$ and $y$, and input trajectories $u$ and $w$ that satisfy~\eqref{Eq: Control}. 
We refer to $h$ as the output of interest, and we refer to $(w,y)$ as internal ports used to model
interconnections between subsystems. The 
external control 
signal $u$ remains available for control design, and the external output 
signal $h$ remains available for performance evaluation, before and after interconnection.


As discussed in the introduction, the input--output behavior of a more refined model can be approximated by that of a simpler model. 
To compare a simpler model with a more refined one in a way that supports controller refinement,
one needs both a function that quantifies their mismatch and a rule that converts external inputs designed
for the simpler model into external inputs for the more refined model.
This idea motivates the notions of storage functions, abstractions, and interface selection rules introduced in~\cite{zamani2017compositional}, which we describe next.

\begin{definition}[{Storage function and abstraction; \cite[Definition 3.1]{zamani2017compositional}}]
\label{Def: Storage_Function}
Consider two open control systems $\Delta^v$, $v\in\{i,j\}$, where
\begin{equation}
\Delta^v = (\mathbb R^{n^v}, \mathbb R^{m^v}, \mathbb R^{p^v},~\mathcal{U}^v, \mathcal{W}^v, f^v,~
\mathbb R^{q^{v}},~\mathbb R^{s^{v}}, h^{v}, y^{v}),    
\end{equation}
and suppose $q^i=q^j=q$ for some~$q \in \mathbb{N}_{> 0}$. 
Let $x^i\in\mathbb R^{n^i}$ and $x^j\in\mathbb R^{n^j}$ denote the states
of~$\Delta^i$ and~$\Delta^j$, respectively.

A continuously differentiable function
\begin{align*}
V^{j,i}:\mathbb R^{n^j}\times \mathbb R^{n^i}\to\mathbb R_{\ge 0}
\end{align*}
is called a \emph{storage function from $\Delta^i$ to $\Delta^j$} if there exist
$\lambda^{j,i}\in\mathcal K_\infty$, $\eta^{j,i}\in\mathcal K_\infty$, and $\rho^{j,i}\in\mathcal K_\infty\cup\{0\}$,
matrices $W^{i,j},W^{j,i}$ and $H^{j,i}$ of appropriate dimensions, and a symmetric matrix
\begin{align*}
X^{j,i}=\begin{bmatrix}[X^{j,i}]_{11} & [X^{j,i}]_{12}\\ [X^{j,i}]_{21} & [X^{j,i}]_{22}\end{bmatrix},
\qquad [X^{j,i}]_{22}\preceq 0,
\end{align*}
such that the following hold\final{:}
\begin{enumerate}
\item \textbf{Output domination:} for all $(x^j,x^i)\in\mathbb R^{n^j}\times\mathbb R^{n^i}$,
\begin{equation}\label{Eq:Simu_Obs}
\lambda^{j,i}\!\left(\left\|h^j(x^j)-h^i(x^i)\right\|\right)\le V^{j,i}(x^j,x^i).
\end{equation}

\item \textbf{Interface and dissipation inequality:}
for every external input $u^i\in\mathcal U^i$ there exists an external input
$u^{j\ast}\in\mathcal U^j$ such that
for all internal inputs $w^i\in\mathcal W^i$ and $w^j\in\mathcal W^j$, we have 
\begin{multline} \label{Eq:Simu_Grad}
\nabla V^{j,i}(x^j,x^i)^\top
\begin{bmatrix}
f^j(x^j,u^{j\ast},w^j)\\[1mm]
f^i(x^i,u^i,w^i)
\end{bmatrix}
\le 
-\eta^{j,i}\!\big(V^{j,i}(x^j,x^i)\big)
+\rho^{j,i}\!\left(\|u^i\|\right)  
\\ 
+
\zeta^{j,i}(x^j,x^i,w^j,w^i)^\top X^{j,i}\,\zeta^{j,i}(x^j,x^i,w^j,w^i),
\end{multline}
where
\begin{align}
\label{eq_zeta}
\zeta^{j,i}(x^j,x^i,w^j,w^i):=
\begin{bmatrix}
W^{i,j}w^j - W^{j,i}w^i\\[1mm]
y^j(x^j)-H^{j,i}y^i(x^i)
\end{bmatrix}.
\end{align}
\end{enumerate}
If such a storage function exists, we write $\Delta^i \preceq \Delta^j$ and say that $\Delta^i$ is an
\emph{abstraction} of $\Delta^j$. With an abuse of notation, we sometimes write $V^{j,i}(x^j,x^i)$ as $V^{j,i}(t)$.
\hfill $\lozenge$
\end{definition}
\begin{remark}[Interface Selection Rule]
\label{remark_Interface}
When a controller is designed on the simpler system $\Delta^i$ but implemented on the more
refined system $\Delta^j$, the input $u^i$ cannot in general be applied directly to $\Delta^j$.
Instead, one needs a rule that converts $u^i$ into an input for $\Delta^j$ while accounting for
the current states of both systems. This role is played by an interface selection rule
$\mathcal I^{j,i}$, which is a measurable map assigning, for each pair of states $(x^j,x^i)$ and
each external input $u^i$, an external input
\begin{align*}
u^{j\ast}=\mathcal I^{j,i}(x^j,x^i,u^i)
\end{align*}
such that the inequality~\eqref{Eq:Simu_Grad} holds. Thus, $\mathcal I^{j,i}$ refines the input
of $\Delta^i$ into an input for $\Delta^j$ through the storage function $V^{j,i}$.
\hfill $\lozenge$
\end{remark}

Under Definition~\ref{Def: Storage_Function}, $V^{j,i}$ is a storage function associated with the pair
$(\Delta^i,\Delta^j)$, and it characterizes the relation $\Delta^i \preceq \Delta^j$. Following
\cite{zamani2017compositional}, a storage function and its interface rule provide a method
for refining inputs from $\Delta^i$ to $\Delta^j$ and a bound on the
difference between the output trajectories that result from applying these inputs. 
Accordingly, we model relationships between open control systems through
storage functions. When a storage function $V^{j,i}$ exists for $(\Delta^i,\Delta^j)$, we write
$\Delta^i \xrightarrow{V^{j,i}} \Delta^j$,   
and we may also write $\Delta^i \to \Delta^j$ when the corresponding storage function and interface are clear from context.

\subsection{Composite Control Systems}\label{Sec_Background_Composite_Control_System}
We form composite control systems by interconnecting open control systems through their
internal ports. 
We briefly review this approach here; for further details
see~\cite{zamani2017compositional,arcak2016networks}. 

Consider $L\in\mathbb N_{>0}$ open control subsystems
\begin{align}
\label{Eq_Control_Sub_Systems}
\Delta_o = (\mathbb{R}^{n_o}, \mathbb{R}^{m_o}, \mathbb{R}^{p_o}, \mathcal{U}_o, \mathcal{W}_o, f_o,
\mathbb{R}^{q_o}, \mathbb{R}^{s_o}, h_o, y_o), \ \ o\in\underline L,
\end{align}
such that
\begin{align*}
     \Delta_o:
    \begin{cases}
\dot x_o(t)=f_o\big(x_o(t),u_o(t),w_o(t)\big),\\
h_o(t)=h_o\big(x_o(t)\big),\\
y_o(t)=y_o\big(x_o(t)\big).
    \end{cases} 
\end{align*}
Let
\begin{align*}
\mathbf{x}&:=[x_1^\top,\ldots,x_L^\top]^\top,\quad
\mathbf{u}:=[u_1^\top,\ldots,u_L^\top]^\top,\quad \\
\mathbf{y}&:=[y_1(x_1)^\top,\ldots,y_L(x_L)^\top]^\top,\quad
\mathbf{w}:=[w_1^\top,\ldots,w_L^\top]^\top.      
\end{align*}

We encode the interconnection pattern among these~$L$ open control systems to construct a composite control system through the following definition.
\begin{definition}[Composite control system {\cite[Def.~4.1]{zamani2017compositional}}]\label{Def_CompositeSystem}
The composite control system induced by $\{\Delta_o\}_{o=1}^L$ is the system
\begin{align*}
\boldsymbol{\Delta}:=(\mathbb R^{\mathbf{n}},\mathbb R^{\mathbf{m}},\boldsymbol{\mathcal U},
\boldsymbol f,\mathbb R^{\mathbf{q}},\boldsymbol h),
\end{align*}
where $\mathbf{n}=\sum_{o=1}^L n_o$, $\mathbf{m}=\sum_{o=1}^L m_o$, and $\mathbf{q}=\sum_{o=1}^L q_o$, and
\begin{align*}
\boldsymbol f(\mathbf x,\mathbf u)&:=
\big[f_1(x_1,u_1,w_1)^\top,\ldots,f_L(x_L,u_L,w_L)^\top\big]^\top,\\
\boldsymbol h(\mathbf x)&:=
\big[h_1(x_1)^\top,\ldots,h_L(x_L)^\top\big]^\top,
\end{align*}
with $(w_1,\ldots,w_L)$ determined by the interconnection constraint
\begin{equation}\label{Eq:InterconnectionConstraint}
\mathbf{w}=\mathbf M\,\mathbf{y},
\end{equation}
where $\mathbf M\in \mathbb R^{\mathbf{p} \times \mathbf{s}}$ is the composition matrix, with $\mathbf{p}=\sum_{o=1}^L p_o$ and $\mathbf{s}=\sum_{o=1}^L s_o$.
\hfill $\lozenge$
\end{definition}

The dimension of $\mathbf{M}$ is determined by the total dimensions of the internal inputs and internal outputs of the $L$ open control subsystems.
We encode the interconnection pattern among these subsystems by the matrix $\mathbf M$. Throughout the paper, by construction $\mathbf M$ has block structure $[\mathbf M]_{ol}\in\mathbb R^{p_o\times s_l}$ and satisfies $[\mathbf M]_{oo}=0$ for all $o,l\in\underline L$, so that no subsystem is connected to itself. Under this construction, $[\mathbf M]_{ol}\in \mathbb R^{p_o\times s_l}$ captures the coupling from the internal outputs of the $l$th subsystem to the internal inputs of the $o^{th}$ subsystem, for all $o,l\in\underline L$ with $o\neq l$.
The composite system $\boldsymbol{\Delta}$ is not an open control system, since it has no internal ports.

We next illustrate the construction of the composition matrix $\mathbf{M}$ for a composite control system using an example.

\begin{example}
Figure~\ref{fig:example_comp_three} shows a composite control system composed of three control subsystems, namely $\Delta_o$ for $o\in \underline 3$.
The internal input and output vectors are given by 
$w = [w_1^\top, w_2^\top, w_3^\top]^\top$ and $y = [y_1^\top, y_2^\top, y_3^\top]^\top$, respectively. Based on Definition~\ref{Def_CompositeSystem},
the dimension of the  composition matrix $\mathbf{M}$ is  $7\times 5$, and $\mathbf{M}$ is given by
\footnotesize
\begin{align*}
\mathbf{M} = 
\begin{bmatrix}
0 & 0 & 0 & 0            &     0\\
0 & 0 & 0 & ([\mathbf{M}]_{31})_{21} &  0\\
([\mathbf{M}]_{21})_{11} & 0 & 0  & 0 & 0\\
0       &        0 & 0 &  0 & ([\mathbf{M}]_{23})_{22}\\
0      &    0  & ([\mathbf{M}]_{32})_{12}      &            0 &  0\\    0    &      ([\mathbf{M}]_{32})_{21} &  0         &           0 &  0\\ 
0       &   0 &  0             &     0  & 0   
\end{bmatrix},
\end{align*}
\normalsize
where the nonzero entries capture the coupling strength from internal outputs to internal inputs. For instance, $([\mathbf{M}]_{21})_{11}\neq 0$ captures the gain from 
the first internal output
 in $\Delta_1$ (which is $(y_1)_1$) to the first internal input in $\Delta_2$ (which is $(w_2)_1$). Additionally, zero entries in $\mathbf{M}$ indicate no interactions between the
 corresponding internal inputs and internal outputs.
 \hfill $\lozenge$
 
\begin{figure}
    \centering
    \includegraphics[width=0.7\linewidth]{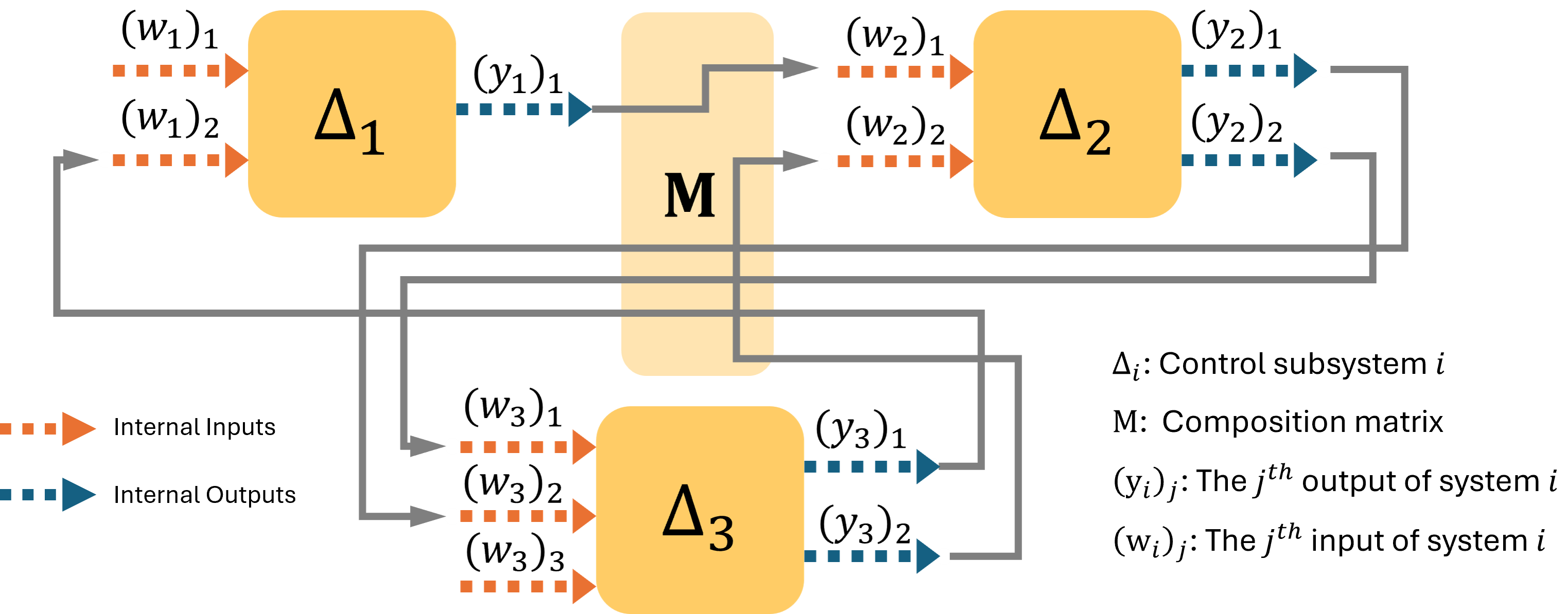}
    \caption{
    A composite control system composed of three control subsystems. \bshe{The composition matrix $\mathbf{M}$ captures the gain when feeding the internal output of one subsystem into the internal input of another. The internal inputs ($(w_1)_1$ or $(w_3)_3$) without any connection are equal to zero. For simplicity, the figure does not show any external inputs and outputs.}}
\label{fig:example_comp_three}
\end{figure}
\end{example}

\subsection{Category Theory} \label{Sec:CategoryTheory}
We recall the basic notions of category and functor.
\begin{definition}[Category {\cite[Def.~3.6]{fong2019invitation}}]\label{Def:Category}
A category $\mathcal C$ consists of:
\begin{enumerate}
\item a collection of objects $\mathrm{Ob}(\mathcal C)$;
\item for every pair of objects $c,d\in\mathrm{Ob}(\mathcal C)$, a set of morphisms
$\mathcal C(c,d)$ (also written $\mathrm{Hom}_{\mathcal C}(c,d)$);
\item for every object $c\in\mathrm{Ob}(\mathcal C)$, an identity morphism
$\mathrm{id}_c\in\mathcal C(c,c)$;
\item for every triple of objects $c,d,e$ and morphisms $f\in\mathcal C(c,d)$ and
$g\in\mathcal C(d,e)$, a composite morphism $g\circ f\in\mathcal C(c,e)$.
\end{enumerate}
These data satisfy:
\begin{enumerate}
\item (identity) $f\circ \mathrm{id}_c=f$ and $\mathrm{id}_d\circ f=f$ for all $f\in\mathcal C(c,d)$;
\item (associativity) $h\circ(g\circ f)=(h\circ g)\circ f$ for all composable $f,g,h$.
\end{enumerate}
\hfill $\lozenge$
\end{definition}
Note that a hom-set $\mathcal C(c,d)$ may be empty. A functor is 
a structure-preserving map between categories.
\begin{definition}[Functor {\cite[Def.~3.35]{fong2019invitation}}]\label{Def:Functor}
Let $\mathcal C$ and $\mathcal D$ be categories. A functor $\Phi:\mathcal C\to\mathcal D$ assigns:
\begin{enumerate}
\item to each object $c\in\mathrm{Ob}(\mathcal C)$ an object $\Phi(c)\in\mathrm{Ob}(\mathcal D)$;
\item to each morphism $f\in\mathcal C(c,d)$ a morphism $\Phi(f)\in\mathcal D(\Phi(c),\Phi(d))$;
\end{enumerate}
such that:
\begin{enumerate}
\item (Preservation of identity) $\Phi(\mathrm{id}_c)=\mathrm{id}_{\Phi(c)}$ for all objects $c$;
\item (Preservation of composition) $\Phi(g\circ f)=\Phi(g)\circ\Phi(f)$ for all composable $f,g$.
\end{enumerate}
\hfill $\lozenge$
\end{definition}
Category theory provides a concise way to describe both systems and the refinement relations between them in a unified framework~\cite{tabuada2004compositional,bakirtzis2021categorical,she2023characterizing}. In this paper, objects will represent control systems, and morphisms represent refinement relations or chains of such relations between systems. A functor then formalizes how these systems and relations can be mapped to another category while preserving the way they compose. For example, in this paper we define a functor that maps an open control system with internal ports to a control system without internal ports by fixing the internal inputs to zero and removing the internal output map, while carrying the corresponding refinement relations to the new setting.

\subsection{Problem Formulation}\label{Sec:ProblemFormulation}
We study families of open control systems and the refinement relations between them induced by storage functions (Definition~\ref{Def: Storage_Function}). In many applications, several models of the same physical process are available and are naturally ordered by fidelity or complexity. In such settings, refinement relations are often available only between models that are close in level of refinement, rather than directly between models separated by several levels in the hierarchy of refinements. 
Moreover, many systems of interest are networked systems obtained by interconnecting subsystems through internal ports.

Motivated by these settings, the problems addressed in this work are as follows:

\begin{problem} \label{prob:1}
Organize open control systems and the refinement relations induced by storage functions into a category in which multi-step refinement across a model hierarchy is represented by composition. 
\end{problem}

\begin{problem} \label{prob:2}
Given a controller designed on a simpler model, quantify the accumulated output mismatch and the resulting closed-loop performance transfer when that controller is implemented on a more refined model through a chain of refinement relations.
\end{problem}

\begin{problem} \label{prob:3}
Formulate a projection 
map to systematically transfer refinement chains for open control systems with internal ports to a corresponding setting without internal ports.
\end{problem}

\begin{problem} \label{prob:4}
Lift subsystem-level refinement chains to interconnected network models for scalable controller refinement across multiple levels in a network model hierarchy.
\end{problem}

\section{Category of Refinement Chains over Model Hierarchies}
\label{Sec_CTRL}
In this section, we study controller refinement for open control systems when refinement relations are available only between models that are close in fidelity within a model hierarchy. Based on the definitions of open control systems in Definition~\ref{Def_Con_Systm}, storage functions in Definition~\ref{Def: Storage_Function}, and the category-theoretic concepts reviewed in Section~\ref{Sec:CategoryTheory}, we formalize multi-step refinement through chains of refinement relations and derive an output mismatch bound between 
the least refined and most refined models in a hierarchy. We then derive a bound on closed-loop performance transfer when a controller designed on a simpler model is implemented on a more 
refined model through the composition of the corresponding refinement relations between them.

\subsection{The Category $\mathsf{OCS}_{\mathrm{chain}}$}
To organize these multi-step refinement relations, 
we introduce a category whose objects are open control systems and whose morphisms are given by
pairwise refinement relations and their compositions. Each elementary step is \final{characterized} by a storage function 
(cf. Definition~\ref{Def: Storage_Function})
together with an associated interface
(cf. Remark~\ref{remark_Interface}), which refines inputs from a simpler model to a more refined model and yields a bound on the mismatch \final{between the corresponding external outputs.} Since such \final{relations} are usually available only between certain pairs of systems, multi-step refinement must be represented by chains of \final{pairwise refinement relations}. 
This need motivates the definition of elementary morphisms and leads to the category $\mathsf{OCS}_{\mathrm{chain}}$.

Let $\mathsf{Sys}$ denote the collection of all open control systems $\Delta$ in Definition~\ref{Def_Con_Systm}.

\begin{definition}[Elementary morphism]
\label{Def_Ele_stor_morph}
Given $\Delta^i,\Delta^j \in \mathsf{Sys}$, an elementary morphism, 
\begin{align}
\alpha^{j,i}:\Delta^i\to\Delta^j
\end{align}
is a pair $\alpha^{j,i}:=(V^{j,i},\mathcal I^{j,i})$, where $V^{j,i}$ is a storage function from $\Delta^i$ to $\Delta^j$ in  Definition~\ref{Def: Storage_Function}, and $\mathcal I^{j,i}$ is the associated interface as in Remark~\ref{remark_Interface}.
\hfill $\lozenge$
\end{definition}

An elementary morphism represents a \final{single-step refinement relation} from one open control system to another. Because multiple refinements can exist between a given pair of systems, elementary morphisms have the structure of a \emph{directed multigraph}, where the vertices are open control systems and edges are refinements.  Since such refinement \final{relations} are typically available only between systems that are close in fidelity or complexity, \final{refinement across multiple levels is represented by chaining these elementary morphisms together.}

\begin{definition}[Chain morphism]
\label{Def_chain_morphism}
Given $\Delta^0,\Delta^r \in \mathsf{Sys}$, a chain morphism $\gamma:\Delta^0\to\Delta^r$ for some $r\in \mathbb{N}_{> 0}$ is a finite sequence of the form 
\begin{align*}
\gamma=(\alpha^{1,0},\alpha^{2,1},\ldots,\alpha^{r,r-1}),
\end{align*}
where, for each $k=0,\ldots,r-1$,
\begin{align*}
\alpha^{k+1,k}:\Delta^k\to\Delta^{k+1},
\end{align*}
so that each adjacent pair of morphisms is composable, \final{i.e., the codomain of $\alpha^{k,k-1}$ coincides with the domain of $\alpha^{k+1,k}$ for each $k=1,\ldots,r-1$.} When $r=0$, the chain is empty and represents the identity morphism on $\Delta^0$.
\hfill $\lozenge$
\end{definition}

A chain morphism defined in Definition~\ref{Def_chain_morphism} represents a multi-step \final{refinement} relation between two open control systems by recording a sequence of elementary morphisms that connect them. This construction is useful if a direct \final{refinement} relation between the initial and final systems in a model hierarchy is unavailable, but \final{such relations} can still be constructed between intermediate systems that are closer in fidelity or complexity. Thus, the chain describes how controller refinement is carried
out step by step across the hierarchy. We now use these chain morphisms to define the category $\mathsf{OCS}_{\mathrm{chain}}$.

\begin{theorem}[Category: $\mathsf{OCS}_{\mathrm{chain}}$]
\label{Thm_CTRL_Path}
Let $\mathsf{Sys}$ denote the collection of all open control systems $\Delta$ in Definition~\ref{Def_Con_Systm}. Define $\mathsf{OCS}_{\mathrm{chain}}$ as follows: 

\emph{Objects:} $\mathrm{Ob}(\mathsf{OCS}_{\mathrm{chain}}):=\mathsf{Sys}$.

\emph{Morphisms:} For any $\Delta^0,\Delta^r\in\mathsf{Sys}$, the morphisms from $\Delta^0$ to $\Delta^r$ are the chain morphisms in Definition~\ref{Def_chain_morphism}, \final{including the empty chain when $\Delta^0=\Delta^r$. In particular, a chain of length one is an elementary morphism in the sense of Definition~\ref{Def_Ele_stor_morph}.}

\emph{Identity:} For each $\Delta\in\mathsf{Sys}$, define $\mathrm{id}_\Delta:=()$, the empty chain.

\emph{Composition:} Given $\gamma:\Delta^0\to\Delta^r$ and $\gamma':\Delta^r\to\Delta^s$, define $\gamma'\circ\gamma$ by concatenation of the two chains.

Then $\mathsf{OCS}_{\mathrm{chain}}$ is a category, i.e., it satisfies all axioms of Definition~\ref{Sec:CategoryTheory}. 
\hfill $\lozenge$
\end{theorem}

\begin{remark}
\label{remark_free_path}
    Observe that chains are formed by taking paths in the directed multigraph of elementary morphisms. Thus, the category $OCS_\mathrm{chain}$ is the \emph{path category} of this directed multigraph: objects are vertices, and morphisms are paths with composition given by concatenation of paths.
\end{remark}

\textcolor{black}{The proof is given in the Appendix.}
Theorem~\ref{Thm_CTRL_Path} provides the structural foundation for the controller refinement problem. In particular, it formalizes the setting in which a direct refinement relation between two open control systems may be unavailable, while a sequence of intermediate refinement steps still connects them. By representing such multi-step refinements as morphisms in the category $\mathsf{OCS}_{\mathrm{chain}}$, Theorem~\ref{Thm_CTRL_Path} allows an entire refinement path from a simpler model to a more refined model to be treated as a single mathematical object. The next theorem then quantifies how output error introduced at each step accumulates along the chain and how it affects closed-loop performance. 

\final{Consider a chain of elementary morphisms 
\begin{equation}\label{eq:chain}
\gamma:\Delta^0 \xrightarrow{\alpha^{1,0}} \Delta^1 \xrightarrow{\alpha^{2,1}} \cdots
\xrightarrow{\alpha^{r,r-1}} \Delta^r,
\end{equation}
where each elementary morphism is of the form
\begin{equation} \label{eq:morphismk}
\alpha^{k,k-1}=(V^{k,k-1},\mathcal I^{k,k-1}), \qquad k\in \underline r.
\end{equation}
According to Definition~\ref{Def: Storage_Function}, $V^{k,k-1}$ is the storage function associated with the pair
$(\Delta^{k-1},\Delta^k)$ and $\mathcal I^{k,k-1}$ is the corresponding interface used to
refine the external input of $\Delta^{k-1}$ into an external input for $\Delta^k$ (cf. Remark~\ref{remark_Interface}). 
For all~$t \geq 0$, given an
external input $u^0(t)$ for $\Delta^0$, the inputs $u^1(t), \ldots, u^r(t)$ are generated recursively along
the chain by
\begin{align}
\label{eq_control_chain}
u^k(t)=\mathcal I^{k,k-1}\bigl(x^k(t),x^{k-1}(t),u^{k-1}(t)\bigr),
\qquad k\in \underline r.
\end{align}
The corresponding state trajectories $(x^0, \ldots, x^r)$ are then generated by the systems
$\Delta^0,\ldots,\Delta^r$, respectively, under these inputs.}

\final{
\begin{theorem}[Controller refinement along chain morphisms]
\label{Thm_Control_Refine}
Let $\gamma$ be a chain morphism in $\mathsf{OCS}_{\mathrm{chain}}$ of the form of~\eqref{eq:chain}, and 
let $(x^0, \ldots, x^r)$ be the corresponding state trajectories generated along the chain as described
in~\eqref{eq_control_chain}. Then, for all $t\ge 0$, the external output mismatch between the source and target systems
$\Delta^0$ and $\Delta^r$ satisfies
\begin{align}
\label{eq:endpoint_bound}
\left\|h^r(x^r(t)) - h^0(x^0(t))\right\|
\le
\sum_{k=1}^r \bigl(\lambda^{k,k-1}\bigr)^{-1}
\Bigl(V^{k,k-1}(x^k(t),x^{k-1}(t))\Bigr),
\end{align}
where, for each $k\in\underline r$, the function $\lambda^{k,k-1}\in\mathcal K_\infty$ is the one
associated with the storage function $V^{k,k-1}$ in Definition~\ref{Def: Storage_Function}.
\hfill $\lozenge$
\end{theorem}
}

\textcolor{black}{The proof is given in the Appendix.}
\final{Theorem~\ref{Thm_Control_Refine} gives an explicit bound on the output mismatch 
between the first system and the last system 
along a chain morphism. To use this result in controller refinement, it remains to bound each intermediate storage function $V^{k,k-1}$ along the corresponding state trajectories. The following corollary provides such bounds directly from the dissipation inequality in Definition~\ref{Def: Storage_Function}.}

\begin{corollary}[Mismatch bounds along chains]
\label{Cor_Control_Bounds}
Consider a chain morphism
\begin{equation}
\gamma:\Delta^0 \xrightarrow{\alpha^{1,0}} \Delta^1 \xrightarrow{\alpha^{2,1}} \cdots
\xrightarrow{\alpha^{r,r-1}} \Delta^r
\end{equation}
from~\eqref{eq:chain}, where
$\alpha^{k,k-1}=(V^{k,k-1},\mathcal I^{k,k-1})$ 
for~$k \in \underline{r}$. 
Let~$x^k$ denote the state trajectory for the system~$\Delta^k$ over~$t \geq 0$. 
For each~$k \in \underline{r}$ consider the function~$\eta^{k,k-1}$ from~\eqref{Eq:Simu_Grad}. 
For each $k\in\underline r$, 
if there exists a constant
$c^{k,k-1}>0$ such that
$\eta^{k,k-1}(s)\ge c^{k,k-1}s$ for all $s\ge 0$,
then the function
\begin{equation}
V^{k,k-1}(t):=V^{k,k-1}(x^k(t),x^{k-1}(t))
\end{equation}
obeys the bound 
\begin{align}
\label{eq:Vk_exp_bound}
V^{k,k-1}(t) &\le e^{-c^{k,k-1}t}V^{k,k-1}(0) 
+ \int_0^t e^{-c^{k,k-1}(t-\tau)} \Big( \rho^{k,k-1}(\|u^{k-1}(\tau)\|) \\ 
& + (\zeta^{k,k-1}(\tau))^\top X^{k,k-1}\zeta^{k,k-1}(\tau) \Big) \,d\tau, \nonumber
\end{align}
\final{where
$\zeta^{k,k-1}(\tau):=
\zeta^{k,k-1}(x^k(\tau),x^{k-1}(\tau),w^k(\tau),w^{k-1}(\tau))$    
 is from~\eqref{eq_zeta}.}
\hfill $\lozenge$
\end{corollary}
\textcolor{black}{The proof is given in the Appendix.}
Theorem~\ref{Thm_Control_Refine} and Corollary~\ref{Cor_Control_Bounds} provide complementary bounds for controller refinement along a chain of refinement relations. The bound in Theorem~\ref{Thm_Control_Refine} shows that the external output mismatch between the endpoint systems $\Delta^0$ and $\Delta^r$ can be quantified by combining the contributions of the intermediate steps along the chain. Thus, even when no direct refinement relation is available between the source model~$\Delta^0$ and the target model~$\Delta^r$, the overall 
output mismatch can still be bounded through the chain.
The bound in Corollary~\ref{Cor_Control_Bounds} further characterizes how the storage function at each step evolves over time under the effects of the external input and the internal coupling term. In particular, when the dissipation term $-\eta^{k,k-1}(V^{k,k-1}(t))$ is sufficiently strong, the mismatch at each step decays up to the forcing induced by the input and interconnection terms. Together, these results show both how the mismatch contributed by each step accumulates along the chain and how it is affected by the underlying dynamics and inputs.

\begin{remark}
Theorem~\ref{Thm_Control_Refine} and Corollary~\ref{Cor_Control_Bounds} support controller refinement across model hierarchies. A controller can be designed on a simpler model and then implemented on a more refined model through a chain of interfaces, while the resulting output mismatch and performance degradation remain explicitly quantifiable. \final{The category-theoretic framework is useful here because it provides a systematic way to organize and compose these multi-step refinement relations into a single chain morphism, thereby making the end-to-end refinement guarantee transparent.}  
\hfill $\lozenge$
\end{remark}

Building on Theorem~\ref{Thm_Control_Refine}, we now derive a closed-loop performance transfer result for controller refinement along the chain.
\begin{corollary}[Performance transfer under refinement]
\label{Cor_CL_Performance_Transfer}
Consider the chain morphism in Theorem~\ref{Thm_Control_Refine}. Suppose a controller is designed for $\Delta^0$ so that the external output of $\Delta^0$ tracks a reference signal $\xi(t)$ with error bounded by a known function, namely,
\begin{align*}
\|h^0(x^0(t)) - \xi(t)\|\le \varepsilon(t),\qquad \textnormal{ for all } t \geq 0
\end{align*}
for some function~$\varepsilon : [0, \infty) \to \mathbb{R}_{\geq 0}$. 
Suppose this controller is implemented on $\Delta^r$ by recursively refining the input along the chain~\eqref{eq:chain} through the interfaces
$\mathcal I^{1,0},\mathcal I^{2,1},\ldots,\mathcal I^{r,r-1}$.
Then, for all $t\ge 0$, the tracking error of $\Delta^r$ satisfies
\begin{align*}
\|h^r(x^r(t)) - \xi(t)\|
\le
\varepsilon(t) + \sum_{k=1}^r (\lambda^{k,k-1})^{-1}\!\bigl(V^{k,k-1}(t)\bigr).
\end{align*}
\hfill $\lozenge$
\end{corollary}

\textcolor{black}{The proof is given in the Appendix.}
Corollary~\ref{Cor_CL_Performance_Transfer} shows that a closed-loop performance guarantee established on a simpler model can be transferred to a more refined model through a chain of refinement relations. The resulting bound combines the original tracking error with an additional term that accounts for the mismatch accumulated along the chain. Thus, the proposed categorical framework yields an explicit overall performance guarantee for controller refinement over model hierarchies.

\section{A Functorial Projection to Control Systems Without Internal Ports}
\label{Sec_CTRLO}
In this section, we study control systems without internal ports and connect the behavior
of such systems to \final{Theorems~\ref{Thm_CTRL_Path} and~\ref{Thm_Control_Refine} and Corollaries~\ref{Cor_Control_Bounds} and~\ref{Cor_CL_Performance_Transfer}} in Section~\ref{Sec_CTRL}. We first define a category whose objects are control systems with only external inputs and outputs and whose morphisms are chains of refinement relations induced by simulation functions. We then define a projection that maps an open control system with both external and internal ports to the corresponding system without internal ports by setting the internal inputs to zero and removing the internal output map. 
Then we show that 
this projection induces a functor from $\mathsf{OCS}_{\mathrm{chain}}$ to the corresponding category of control systems without internal ports.

\subsection{A Category of Control Systems Without Internal Ports}
\begin{definition}[Control system without internal ports]
\label{Def_Con_Systm_without_int}
A control system $\bar{\Delta}$ is a tuple
\begin{align*}
 \bar{\Delta} = (\mathbb R^n, \mathbb R^m, \mathcal{U}, \bar f, \mathbb R^{q}, h),
\end{align*}
where $\mathbb R^n$ is the state space, $\mathbb R^m$ is the external input space, and $\mathbb R^{q}$ 
is the external output space, and
\begin{itemize}
\item $\mathcal{U}$ is a \final{non-empty} subset of the set of all measurable functions of time from open intervals
in $\mathbb R$ to $\mathbb R^m$,
\item $\bar f: \mathbb R^n \times \mathbb R^m \rightarrow \mathbb R^n$ is a continuous map satisfying
the following Lipschitz assumption: for every compact set $D\subseteq \mathbb R^n$, there exists a
constant $\tau\in \mathbb R_{>0}$ such that
\begin{align*}
\|\bar f(x^1,u)-\bar f(x^2,u)\|\leq \tau \|x^1-x^2\|
\end{align*}
for all $x^1, x^2\in D$ and all $u\in \mathbb R^m$,
\item $h: \mathbb R^n\rightarrow \mathbb R^{q}$ is the  \final{continuous} external output map.
\end{itemize}
\hfill $\lozenge$
\end{definition}

As introduced in~\cite[Definition 3.2]{zamani2017compositional}, similar to storage functions for
open control systems with internal ports, one can define simulation functions for control systems
without internal ports.
\begin{definition}[{Simulation function; \cite[Definition 3.2]{zamani2017compositional}}]
\label{Def: Simulation_Function}
Consider two control systems without internal ports of the form
\begin{align*}
\bar{\Delta}^v=(\mathbb R^{n^v},\mathbb R^{m^v},\mathcal U^v,\bar f^v,\mathbb R^{q^v},h^v),
\qquad v\in\{i,j\},
\end{align*}
with dynamics
\begin{align*}
\bar{\Delta}^v:\qquad
\begin{cases}
\dot x^v(t)=\bar f^v(x^v(t),u^v(t)) \\
y^v(t)=h^v(x^v(t)) 
\end{cases}
\qquad v\in\{i,j\}.
\end{align*}
We say that $\bar{\Delta}^i$ is an abstraction of $\bar{\Delta}^j$ if there exists a continuously
differentiable simulation function
\begin{align*}
\bar V^{j,i}: \mathbb R^{n^j}\times \mathbb R^{n^i} \rightarrow \mathbb R_{\geq0}
\end{align*}
and functions $\lambda^{j,i}, \eta^{j,i}\in \mathcal{K}_{\infty}$ and
$\rho^{j,i}\in \mathcal{K}_{\infty}\cup \{0\}$ such that the following conditions hold: 
\begin{enumerate}
\item \textbf{Output domination:}
for all $x^i\in \mathbb R^{n^i}$ and $x^j \in \mathbb R^{n^j}$,
\begin{align}
\label{Eq:SimFunctionObs}
\lambda^{j,i}(\|h^{j}(x^j)- h^{i}(x^i)\|)\leq \bar {V}^{j,i}(x^j, x^i).
\end{align}

\item \textbf{Interface and dissipation inequality:}
for every external input $u^i\in \mathcal U^i$, there exists an external input
$u^{j\ast}\in \mathcal U^j$ such that
\begin{equation} \label{Eq_Grad_Composite}
\nabla \bar V^{j,i}(x^j, x^i)^\top
    \begin{bmatrix}
        \bar f^j(x^j,u^{j\ast})\\
        \bar f^i(x^i,u^i)
    \end{bmatrix} 
\leq -\eta^{j,i} (\bar V^{j,i}(x^j, x^i)) +\rho^{j,i} (\|u^i\|).
\end{equation}
\end{enumerate}
If such a simulation function exists, we write $\bar{\Delta}^i \preceq \bar{\Delta}^j$ and say that
$\bar{\Delta}^i$ is an abstraction of $\bar{\Delta}^j$.
\hfill $\lozenge$
\end{definition}

Storage functions in Definition~\ref{Def: Storage_Function} are defined for open control systems with internal and external ports, whereas simulation functions in Definition~\ref{Def: Simulation_Function} are defined for control systems with only external ports. We next define a category for this latter class of systems in which morphisms are given by chains of refinement relations induced by simulation functions, which parallels the construction of $\mathsf{OCS}_{\mathrm{chain}}$ in Section~\ref{Sec_CTRL}.

Let $\mathsf{Sys}_O$ denote the collection of all control systems without internal ports in the sense of 
Definition~\ref{Def_Con_Systm_without_int}.

\begin{definition}[Elementary simulation morphism]
\label{Def_Ele_Sim_Morph}
Given $\bar{\Delta}^i,\bar{\Delta}^j\in\mathsf{Sys}_O$, an elementary simulation morphism
\begin{align*}
\bar{\alpha}^{j,i}:=(\bar V^{j,i},\bar{\mathcal I}^{j,i}):\bar{\Delta}^i\to\bar{\Delta}^j
\end{align*}
consists of a simulation function $\bar V^{j,i}$ from $\bar{\Delta}^i$ to $\bar{\Delta}^j$ in 
the sense of Definition~\ref{Def: Simulation_Function}, together with the associated interface selection
rule $\bar{\mathcal I}^{j,i}$ as defined in Remark~\ref{remark_Interface}. 
\hfill $\lozenge$
\end{definition}

An elementary simulation morphism represents a single-step refinement relation between two systems without internal ports, which is characterized by a simulation function. Since such relations may exist only between systems that are close in fidelity or complexity, we again represent multi-step refinement through chains of elementary simulation morphisms.

\begin{definition}[Simulation chain morphism]
\label{Def_chain_morphism_O}
Given $\bar{\Delta}^0,\bar{\Delta}^r\in\mathsf{Sys}_O$, a simulation chain morphism
\begin{align*}
\bar{\gamma}:\bar{\Delta}^0\to\bar{\Delta}^r
\end{align*}
is, for some $r\in\mathbb N_0$, a finite sequence of the form 
\begin{align*}
\bar{\gamma}=(\bar{\alpha}^{1,0},\bar{\alpha}^{2,1},\ldots,\bar{\alpha}^{r,r-1}),
\end{align*}
where
\begin{align*}
\bar{\alpha}^{k+1,k}:\bar{\Delta}^k\to\bar{\Delta}^{k+1}, \qquad k=0,\ldots,r-1,
\end{align*}
so that each consecutive pair 
of functions in~$\bar{\gamma}$ 
is composable. When $r=1$, the chain consists of a single
elementary simulation morphism, and when $r=0$, it is the empty chain.
\hfill $\lozenge$
\end{definition}
A simulation chain morphism represents a multi-step \final{refinement} relation
between two control systems without internal ports. This notion provides the morphisms needed to define a
category whose composition is given by concatenation of compatible chains.

\begin{theorem}[Category: $\mathsf{CS}_{\mathrm{chain}}$]
\label{Def:CTRL-O-path}
Let~$\mathsf{Sys}_O$ denote the collection of all open control systems
in~$\Delta_0$ in Definition~\ref{Def_Con_Systm_without_int}. 
Define~$\mathsf{CS}_{\mathrm{chain}}$ as follows: 

\emph{Objects:}
\begin{align*}
\mathrm{Ob}(\mathsf{CS}_{\mathrm{chain}}):=\mathsf{Sys}_O.
\end{align*}

\emph{Morphisms:}
for any $\bar{\Delta}^0,\bar{\Delta}^r\in\mathsf{Sys}_O$, the morphisms from
$\bar{\Delta}^0$ to $\bar{\Delta}^r$ are the simulation chain morphisms in
Definition~\ref{Def_chain_morphism_O}, including the empty chain when
$\bar{\Delta}^0=\bar{\Delta}^r$. In particular, a chain of length one is simply an
elementary simulation morphism.

\emph{Identity:}
for each $\bar{\Delta}\in\mathsf{Sys}_O$, define
$\mathrm{id}_{\bar{\Delta}}:=()$, which is 
the empty chain.

\emph{Composition:}
given $\bar{\gamma}:\bar{\Delta}^0\to\bar{\Delta}^r$ and
$\bar{\gamma}':\bar{\Delta}^r\to\bar{\Delta}^s$, define $\bar{\gamma}'\circ\bar{\gamma}$ to be the concatenation of the two chains.

Then $\mathsf{CS}_{\mathrm{chain}}$ is a category.
\hfill $\lozenge$
\end{theorem}

The proof is analogous to that of Theorem~\ref{Thm_CTRL_Path} and is therefore omitted. 
Theorem~\ref{Thm_CTRL_Path} defines the category $\mathsf{OCS}_{\mathrm{chain}}$ of systems with internal ports,
and 
Theorem~\ref{Def:CTRL-O-path} defines the category 
$\mathsf{CS}_{\mathrm{chain}}$ 
for control systems without internal ports. Introducing these two categories makes it possible to study \final{refinement} relations and their compositions in a unified way across both settings. Rather than analyzing $\mathsf{CS}_{\mathrm{chain}}$ independently, the next section shows that it can be studied through the structure of $\mathsf{OCS}_{\mathrm{chain}}$ via a functorial projection. As a result, controller refinement results developed for open control systems can be transferred automatically and systematically to systems without internal ports.
\subsection{A Functorial Projection From Open Control Systems}
We next define a projection from open control systems with internal ports to control systems without internal ports. 
\begin{definition}[Object map $\Phi$]
\label{Def:PhiObject}
Define
\begin{align*}
\Phi:\mathsf{Sys}\to\mathsf{Sys}_O
\end{align*}
as follows. For any open control system
\begin{align*}
\Delta=(\mathbb R^n,\mathbb R^m,\mathbb R^p,\mathcal U,\mathcal W,f,\mathbb R^q,\mathbb R^s,h,y)
\in\mathsf{Sys},
\end{align*}
set
\begin{align*}
\Phi(\Delta):=\bar{\Delta}:=(\mathbb R^n,\mathbb R^m,\mathcal U,\bar f,\mathbb R^q,h)\in\mathsf{Sys}_O,
\end{align*}
where
\begin{align*}
\bar f(x,u):=f(x,u,0) \ \text{for all} \  (x,u)\in\mathbb R^n\times\mathbb R^m.
\end{align*}
\hfill $\lozenge$
\end{definition}
At the object level, the map $\Phi$ removes the internal ports by fixing the internal input to zero and discarding the internal output map. We next show that an elementary morphism between two open control systems can be carried to an elementary morphism between the corresponding projected systems without internal ports. The key connection is that, once the internal port terms are removed, the storage function inequalities reduce to the simulation function inequalities for the projected systems.

\begin{lemma}[Projection of an elementary morphism]
\label{Lem:PhiElementary}
Consider the category $\mathsf{OCS}_{\mathrm{chain}}$ in Theorem~\ref{Thm_CTRL_Path}. 
Let
\begin{equation} \label{eq:alphajidef}
\alpha^{j,i}:=(V^{j,i},\mathcal I^{j,i}):\Delta^i\to\Delta^j
\end{equation}
be an elementary morphism in $\mathsf{OCS}_{\mathrm{chain}}$, where for $v\in\{i,j\}$,
\begin{align*}
\Delta^v=(\mathbb R^{n^v},\mathbb R^{m^v},\mathbb R^{p^v},\mathcal U^v,\mathcal W^v,f^v,
\mathbb R^{q^v},\mathbb R^{s^v},h^v,y^v).
\end{align*}
Let
$\bar{\Delta}^i:=\Phi(\Delta^i)$ and $\bar{\Delta}^j:=\Phi(\Delta^j)$
be the corresponding projected control systems without internal ports. Then the function $V^{j,i}$
in~\eqref{eq:alphajidef}
satisfies the conditions
for being a simulation function 
in Definition~\ref{Def: Simulation_Function} for the pair
$(\bar{\Delta}^i,\bar{\Delta}^j)$. Consequently,
\begin{equation} \label{eq:lambdajidef}
\bar{\alpha}^{j,i}:=(V^{j,i},\mathcal I^{j,i}):\bar{\Delta}^i\to\bar{\Delta}^j
\end{equation}
is an elementary simulation morphism in $\mathsf{CS}_{\mathrm{chain}}$. 
\hfill $\lozenge$
\end{lemma}

\textcolor{black}{The proof is given in the Appendix.}
Lemma~\ref{Lem:PhiElementary} shows that an \final{elementary morphism} in $\mathsf{OCS}_{\mathrm{chain}}$ \final{induces an elementary simulation morphism} between the corresponding projected control systems without internal ports
in~$\mathsf{CS}_{\mathrm{chain}}$. 
We next extend this projection from elementary morphisms to chains of morphisms, which in turn leads to the projection at the category level.
\final{Consider an elementary morphism
\begin{align*}
\alpha^{j,i}=(V^{j,i},\mathcal I^{j,i}):\Delta^i\to\Delta^j
\end{align*}
in $\mathsf{OCS}_{\mathrm{chain}}$ as defined in Definition~\ref{Def_Ele_stor_morph}. 
Lemma~\ref{Lem:PhiElementary} shows that $V^{j,i}$ is a simulation
function from $\Phi(\Delta^i)$ to $\Phi(\Delta^j)$.}

\begin{definition}[Morphism map $\Phi$]
\label{Def:PhiMorphism}
We  define
\begin{align*}
\Phi(\alpha^{j,i}):=\bar{\alpha}^{j,i}:=(V^{j,i},\mathcal I^{j,i}),
\end{align*}
viewed now as an elementary simulation morphism
\begin{align*}
\bar{\alpha}^{j,i}:\Phi(\Delta^i)\to\Phi(\Delta^j)
\end{align*}
in $\mathsf{CS}_{\mathrm{chain}}$.

We extend $\Phi$ to chain morphisms by acting on each elementary morphism
within a chain morphism. 
In particular, if
\begin{align*}
\gamma=(\alpha^{1,0},\alpha^{2,1},\ldots,\alpha^{r,r-1})
\end{align*}
is a chain morphism in $\mathsf{OCS}_{\mathrm{chain}}$ as defined in Definition~\ref{Def_chain_morphism}, then 
\begin{align*}
\Phi(\gamma):=
(\Phi(\alpha^{1,0}),\Phi(\alpha^{2,1}),\ldots,\Phi(\alpha^{r,r-1})).
\end{align*}
\hfill $\lozenge$
\end{definition}
With $\Phi$ now defined on both objects and morphisms, it remains to verify that it preserves identities and composition. The following theorem shows that $\Phi$ indeed does so and 
therefore it defines a functor from $\mathsf{OCS}_{\mathrm{chain}}$ to $\mathsf{CS}_{\mathrm{chain}}$.

\begin{theorem}[Functorial projection]
\label{Thm:PhiFunctor}
Let $\Phi : \mathsf{OCS}_{\mathrm{chain}} \to \mathsf{CS}_{\mathrm{chain}}$ 
be defined on objects as in Definition~\ref{Def:PhiObject} and on morphisms as in
Definition~\ref{Def:PhiMorphism}. Then $\Phi$ defines a functor. 
In particular, for every object $\Delta\in\mathsf{Sys}$,
\begin{align*}
\Phi(\mathrm{id}_{\Delta})=\mathrm{id}_{\Phi(\Delta)},
\end{align*}
and for every pair of composable chain morphisms $\gamma,\gamma'$ in $\mathsf{OCS}_{\mathrm{chain}}$, it satisfies
\begin{align*}
\Phi(\gamma'\circ\gamma)=\Phi(\gamma')\circ\Phi(\gamma).
\end{align*}
\hfill $\lozenge$
\end{theorem}
The functor $\Phi$ provides a systematic way to project open control systems with internal and external ports to control systems without internal ports while preserving \final{chains of refinement relations.} 
This \final{property} shows that the chain-based refinement framework developed for open control systems is compatible with the simpler setting involving 
systems with 
only external inputs and outputs. In particular, a chain morphism between interconnected models can be carried to a corresponding chain morphism between their projected systems without changing the chain structure.

Consequently, the refinement results of Section~\ref{Sec_CTRL} remain available after projection. Once a chain morphism in $\mathsf{OCS}_{\mathrm{chain}}$ is mapped to a chain
morphism in $\mathsf{CS}_{\mathrm{chain}}$, the same analysis can be used to
relate the projected endpoint systems. The following corollaries formalize this consequence.

\begin{corollary}[Projected controller performance transfer]
\label{Cor:PhiTransfer_1}
Let $\gamma:\Delta^0\to\Delta^r$ be a chain morphism in $\mathsf{OCS}_{\mathrm{chain}}$ of the form of~\eqref{eq:chain}, 
and let
$\Phi(\gamma):\Phi(\Delta^0)\to\Phi(\Delta^r)$ be its image under the functor~$\Phi$
from Theorem~\ref{Thm:PhiFunctor}. 
Then the controller refinement and performance transfer results in
Theorem~\ref{Thm_Control_Refine} and Corollary~\ref{Cor_CL_Performance_Transfer} hold for the projected chain $\Phi(\gamma)$ in
$\mathsf{CS}_{\mathrm{chain}}$, with the functions along the chain viewed as simulation functions for the projected systems. \hfill $\lozenge$
\end{corollary}

\textcolor{black}{The proof is given in the Appendix.}
\begin{corollary}[Projected explicit refinement bound]
\label{Cor:PhiExplicit_2}
Let
\begin{align*}
\gamma:\Delta^0\to\Delta^r
\end{align*}
be a chain morphism in $\mathsf{OCS}_{\mathrm{chain}}$, and let
\begin{align*}
\Phi(\gamma):\Phi(\Delta^0)\to\Phi(\Delta^r)
\end{align*}
be its image in $\mathsf{CS}_{\mathrm{chain}}$. For each $k\in\underline r$, let 
$\bar{\alpha}^{k, k-1} = (V^{k,k-1}, \mathcal{I}^{k,k-1})$ 
denote the $k^{th}$ elementary simulation morphism in $\Phi(\gamma)$, and 
for~$t \geq 0$ write
\begin{align*}
V^{k,k-1}(t):=V^{k,k-1}(x^k(t),x^{k-1}(t)).
\end{align*}
For each~$k \in \underline{r}$, consider the function~$\eta^{k,k-1}$ from~\eqref{Eq:Simu_Grad}. 
If there exists a constant
$c^{k,k-1}>0$ such that, for all $s\ge 0$,
\begin{align*}
\eta^{k,k-1}(s)\ge c^{k,k-1}s, 
\end{align*}
then, for all $t\ge 0$, we have 
\begin{equation}
V^{k,k-1}(t)
\le e^{-c^{k,k-1}t}V^{k,k-1}(0)  +
\int_0^t e^{-c^{k,k-1}(t-\tau)}
\rho^{k,k-1}\big(\|u^{k-1}(\tau)\|\big)\,d\tau.
\end{equation}
\hfill $\lozenge$
\end{corollary}

For each $k\in\underline r$, the function $V^{k,k-1}$ satisfies the corresponding dissipation inequality for the projected systems, without the terms from the internal ports. The stated bound then follows by the same argument as in Corollary~\ref{Cor_Control_Bounds}, with the terms from the internal ports removed after projection, and the proof details are therefore omitted.

Corollaries~\ref{Cor:PhiTransfer_1} and~\ref{Cor:PhiExplicit_2} show that the controller refinement and performance transfer results for open control systems with internal ports carry over to the corresponding projected systems without internal ports. In this way, the functor $\Phi$ allows the same chain of refinement relations to be studied in both settings. 

This section introduced the category $\mathsf{CS}_{\mathrm{chain}}$ for control systems without internal ports, with morphisms given by chains of refinement relations induced by simulation functions. It then defined a functorial projection from $\mathsf{OCS}_{\mathrm{chain}}$ to $\mathsf{CS}_{\mathrm{chain}}$ by removing the internal ports and showed that this projection preserves the chain structure. Consequently, the controller refinement framework developed in Section~\ref{Sec_CTRL} remains available in the setting 
of systems without internal ports as well.


\section{\final{From Subsystem Storage Functions to Morphisms Between Composite Systems}}\label{Sec:NetworkLifting}

In this section, we study how storage functions between open control subsystems
can be combined to produce simulation functions between the corresponding composite control
systems. Since a composite control system in Definition~\ref{Def_CompositeSystem} has only
external inputs and outputs, the resulting relation 
between the composite systems 
is naturally expressed as a simulation function
in $\mathsf{CS}_{\mathrm{chain}}$. This construction provides a scalable route for controller refinement,
since relations between large network models can be obtained by combining subsystem-level
storage functions rather than constructing a simulation function directly at the network level.
After such relations are established between successive network models in a hierarchy,
the controller refinement and performance transfer results from Theorem~\ref{Thm_Control_Refine} and
Corollary~\ref{Cor_CL_Performance_Transfer} can be applied across that hierarchy.

In this section, we use superscripts to denote different models of the same physical process across
a model hierarchy, and we use subscripts to denote different physical processes that are composed to
form a networked system, \textcolor{black}{as shown in Figure~\ref{fig_comp_ex_2}}. 
Thus, for each fixed subscript, the superscript labels different models
in the hierarchy of the same open control subsystem, while for each fixed superscript, the
subscript labels the open control subsystems that form the corresponding composite control
system at the level of the hierarchy indexed by the superscript. 
\begin{figure}
    \centering
    \includegraphics[width=1\linewidth]{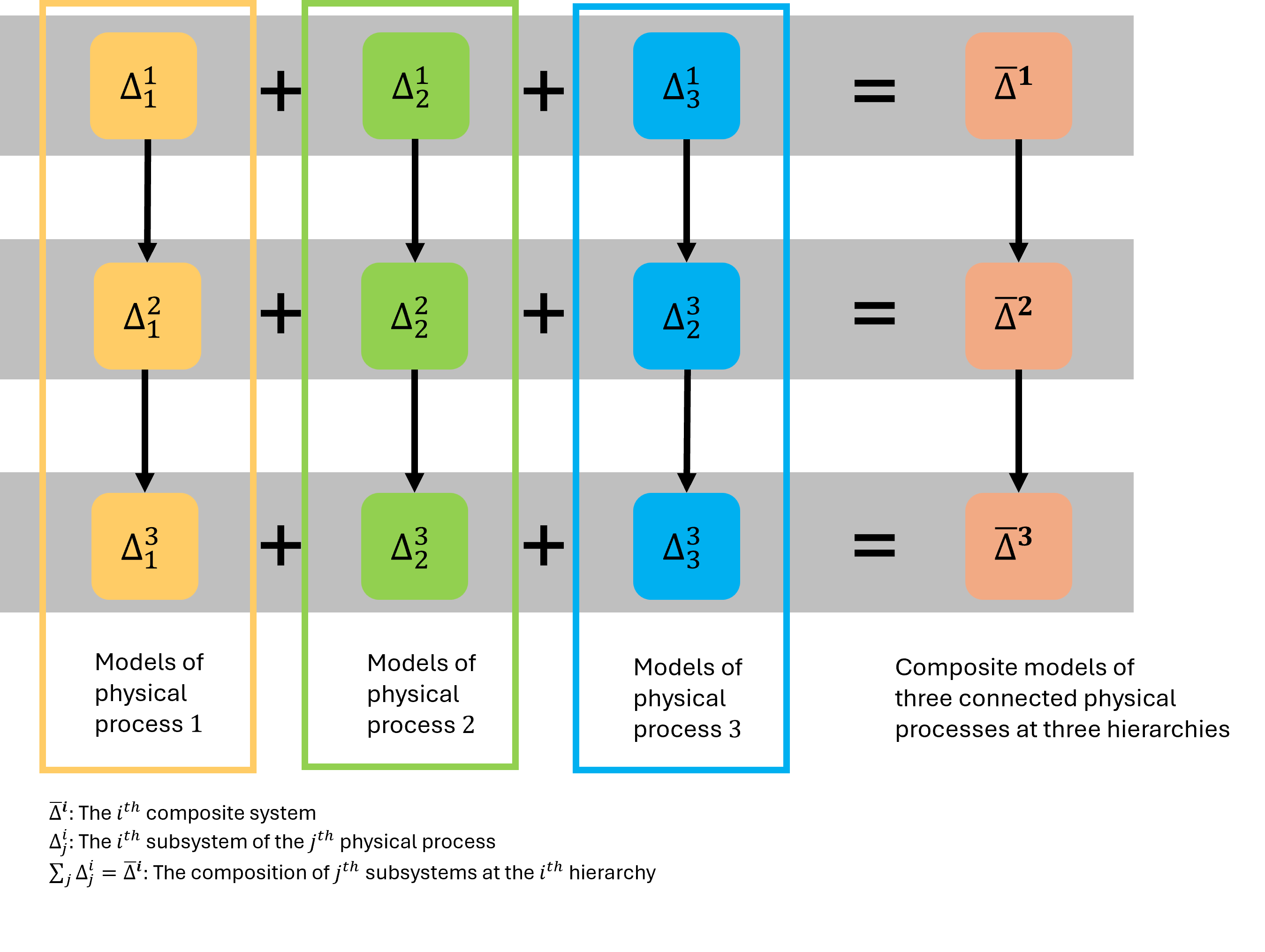}
    \caption{\textcolor{black}{Example of three composite systems. A composite physical process is composed of three physical processes. The first three columns represent models of these three physical processes. 
     Each physical process is represented by three subsystem models at three different levels of the hierarchy. For example, the first column contains the models $\Delta^1_1\rightarrow\Delta^2_1\rightarrow\Delta^3_1$ for physical process~1. Each row forms a composite system by composing the subsystem models of the three physical processes at the $i^{th}$ level of the hierarchy, where $i\in\underline{3}$. Thus, the three rows yield the three composite systems shown in the last column.}}
    \label{fig_comp_ex_2}
\end{figure}
\subsection{Composite Control Systems from Collections of Open Control Subsystems}\label{Sec:CompositeSetup}
In order to distinguish composite control systems 
from the open control systems that are composed to construct
them, we refer to the latter as open control subsystems. We consider
$L\in\mathbb N_{>0}$ open control subsystems at each of $N \in \mathbb N_{>0}$ levels in the model
hierarchy. For each $o\in\underline{L}$ and $i\in\underline{N}$, let
\begin{align}
\label{Eq:subsystem_group_notation}
\Delta_o^i =
(\mathbb R^{n_o^i},\mathbb R^{m_o^i},\mathbb R^{p_o^i},\mathcal U_o^i,\mathcal W_o^i,f_o^i,
\mathbb R^{q_o^i},\mathbb R^{s_o^i},h_o^i,y_o^i)
\end{align}
denote the $o^{th}$ open control subsystem at the $i^{th}$ level of the model hierarchy.
Each such open control subsystem satisfies 
Definition~\ref{Def_Con_Systm}. For each fixed $i\in\underline{N}$, the collection
$\{\Delta_o^i\}_{o=1}^L$ represents the open control subsystems that are interconnected to form
the composite control system at the $i^{th}$ level of the model hierarchy.

For each fixed $i$, we compose the open control subsystems $\{\Delta_o^i\}_{o=1}^L$ using the
composition rule in Definition~\ref{Def_CompositeSystem}. We denote the resulting $i^{th}$
composite control system by
\begin{align*}
\bar{\boldsymbol{\Delta}}^i=
(\mathbb R^{\mathbf{n}^i},\mathbb R^{\mathbf{m}^i},\boldsymbol{\mathcal U}^i,
\boldsymbol f^i,\mathbb R^{\mathbf{q}^i},\boldsymbol h^i),    
\end{align*}
where
\begin{align*}
\mathbf{n}^i=\sum_{o=1}^L n_o^i,\qquad
\mathbf{m}^i=\sum_{o=1}^L m_o^i,\qquad
\mathbf{q}^i=\sum_{o=1}^L q_o^i.    
\end{align*}

Its state and external input are given by the stacked vectors
\begin{align*}
\boldsymbol x^i=[(x_1^i)^\top,\ldots,(x_L^i)^\top]^\top,\qquad
\boldsymbol u^i=[(u_1^i)^\top,\ldots,(u_L^i)^\top]^\top,    
\end{align*}
and its dynamics and external output map are induced by the subsystem dynamics and the
corresponding composition matrix, as in Definition~\ref{Def_CompositeSystem}. Following the
notation in Section~\ref{Sec_Background_Composite_Control_System}, we also use
\begin{align*}
\bar{\boldsymbol{\Delta}}^i:
\begin{cases}
\dot {\mathbf{x}}^i(t)=\boldsymbol{f}^i\big(\mathbf{x}^i(t),\mathbf{u}^i(t))\\
\boldsymbol{h}^i(t)=\boldsymbol{h}^i\big(\mathbf{x}^i(t)\big)
    \end{cases} 
\end{align*}
to denote the composite control system at the $i^{th}$ \final{level of the hierarchy,} for all $i\in \underline N$.

\final{For each fixed $i$, the internal inputs and outputs of the open control subsystems
$\{\Delta_o^i\}_{o=1}^L$ are connected through the interconnection rule in
Definition~\ref{Def_CompositeSystem}. After this composition, 
the internal ports are no longer free variables of the overall model. As a result,
$\bar{\boldsymbol{\Delta}}^i$ is a control system with only external inputs and outputs, and it 
therefore belongs to $\mathsf{Sys}_O$.} 

Based on this formulation, we can compare two composite control systems
$\bar{\boldsymbol{\Delta}}^i$ and $\bar{\boldsymbol{\Delta}}^j$ in $\mathsf{Sys}_O$
through their corresponding
pairs of open control subsystems $\Delta_o^i$ and $\Delta_o^j$ for all $o\in\underline{L}$ in $\mathsf{Sys}$. We next show
that, under suitable compatibility conditions, the storage functions between corresponding subsystems can
be combined into a simulation function between the corresponding composite systems.

\subsection{\final{From Subsystem Storage Functions to A Composite System Simulation Function}}
\label{Sec:LiftThm}

\final{We first recall a compositional result from
\cite[Theorem~4.2]{zamani2017compositional}.
That theorem gives conditions under which storage functions between corresponding open control
subsystems can be combined into a simulation function between the corresponding composite
control systems.}

\begin{proposition}[{\cite[Theorem~4.2]{zamani2017compositional}}]
\label{Thm:CompositionalSimulation}
\final{Consider two composite control systems
$\bar{\boldsymbol{\Delta}}^i$ and $\bar{\boldsymbol{\Delta}}^j$ obtained by composing the collections of open control subsystems
$\{\Delta_o^i\}_{o=1}^L$ and $\{\Delta_o^j\}_{o=1}^L$, respectively, 
according to Definition~\ref{Def_CompositeSystem}. Let $M^i$ and $M^j$ denote their
corresponding composition matrices.}
\final{Suppose that, for every $o\in\underline L$, there exists a storage function
\begin{align*}
V_o^{j,i}:
\mathbb R^{n_o^j}\times\mathbb R^{n_o^i}
\rightarrow\mathbb R_{\geq 0}
\end{align*}
from $\Delta_o^i$ to $\Delta_o^j$ in the sense of
Definition~\ref{Def: Storage_Function}.}
Let~$W_o^{j,i}$, $W_o^{i,j}$, $H_o^{j,i}$ be the matrices from Definition~\ref{Def: Storage_Function} associated with~$V^{j,i}_o$. 

\final{Suppose that there exist weights $\mu_o^{j,i}\geq 0$ for $o\in\underline L$ such that
\begin{align}
\label{Eq:Matrix_Condition_1}
(\Theta^j)^\top X^{j,i}\Theta^j\preceq 0
\end{align}
and
\begin{align}
\label{Eq:Matrix_Condition_2}
W^{i,j}M^jH^{j,i}=W^{j,i}M^i,
\end{align}
where
\begin{align*}
\Theta^j
:=
\begin{bmatrix}
W^{i,j}M^j\\
I
\end{bmatrix},
\end{align*}
\begin{multline}
W^{j,i}
:=
\operatorname{diag}
\bigl(W_1^{j,i},\ldots,W_L^{j,i}\bigr), \quad 
W^{i,j}
:=
\operatorname{diag}
\bigl(W_1^{i,j},\ldots,W_L^{i,j}\bigr), \\ \quad
H^{j,i}
:=
\operatorname{diag}
\bigl(H_1^{j,i},\ldots,H_L^{j,i}\bigr),
\end{multline}
and
\begin{align*}
X^{j,i}
:=
\begin{bmatrix}
X_{11}^{j,i} & X_{12}^{j,i}\\
X_{21}^{j,i} & X_{22}^{j,i}
\end{bmatrix},
\end{align*}
with
\begin{align*}
X_{\ell m}^{j,i}
:=
\operatorname{diag}\bigl(
\mu_1^{j,i}[X_1^{j,i}]_{\ell m},
\ldots,
\mu_L^{j,i}[X_L^{j,i}]_{\ell m}
\bigr),
\qquad \ell,m\in\{1,2\}.
\end{align*}
}

\final{Then
\begin{align}
\label{Eq_Simulation_Func_Weighted_Sum}
\boldsymbol{V}^{j,i}(\boldsymbol{x}^j,\boldsymbol{x}^i)
:=
\sum_{o=1}^{L}
\mu_o^{j,i}V_o^{j,i}(x_o^j,x_o^i)
\end{align}
is a simulation function from
$\bar{\boldsymbol{\Delta}}^i$ to
$\bar{\boldsymbol{\Delta}}^j$ in the sense of
Definition~\ref{Def: Simulation_Function}.
\hfill $\lozenge$}
\end{proposition}

\final{For all $o\in\underline L$, \eqref{Eq:Matrix_Condition_2} 
ensures that
the  transformations at the subsystem level
associated with the storage functions are compatible with the composition matrices
$M^i$ and $M^j$. Meanwhile, \eqref{Eq:Matrix_Condition_1} ensures that, after the
subsystem dissipation inequalities are weighted and summed, the contribution of the
internal interconnections is nonpositive. The coefficients $\mu_o^{j,i}$ determine how
the subsystem storage functions are combined in
\eqref{Eq_Simulation_Func_Weighted_Sum}.}

\begin{remark}
Suppose we have two collections of open control subsystems that are composed to form two composite systems. 
 To distinguish functions associated with individual subsystems from those
associated with the composite systems, we use 
$V_o^{j,i}$ to denote a storage function between the $o^{th}$ pair of 
open control subsystems, whereas $\boldsymbol V^{j,i}$ denotes the simulation
function between the corresponding composite control systems. The function
$\boldsymbol V^{j,i}$ is a simulation function in the sense of
Definition~\ref{Def: Simulation_Function}, and the boldface notation indicates that it is constructed for
the composite systems.
\hfill $\lozenge$
\end{remark}

\final{
\begin{proposition}[Category-theoretic form of
Proposition~\ref{Thm:CompositionalSimulation}]
\label{Prop:Composite_Control}
Under the assumptions of Proposition~\ref{Thm:CompositionalSimulation}, suppose that the simulation
function
$\boldsymbol{V}^{j,i}$ is accompanied by an interface selection rule
$\boldsymbol{\mathcal I}^{j,i}$ satisfying Definition~\ref{Def_Ele_Sim_Morph}. 
Then,
\begin{align*}
\bar{\boldsymbol{\alpha}}^{j,i}
:=
\bigl(\boldsymbol{V}^{j,i},\boldsymbol{\mathcal I}^{j,i}\bigr):
\bar{\boldsymbol{\Delta}}^i
\longrightarrow
\bar{\boldsymbol{\Delta}}^j
\end{align*}
defines an elementary simulation morphism in
$\mathsf{CS}_{\mathrm{chain}}$. \hfill $\lozenge$
\end{proposition}}
The proof is given in the Appendix. 
\final{We now connect the construction of composite control systems with the
chain-based refinement framework developed in Sections~\ref{Sec_CTRL} and~\ref{Sec_CTRLO}.
For each pair of adjacent levels in a network model hierarchy, the compositional result in
Proposition~\ref{Thm:CompositionalSimulation} combines storage functions between corresponding
open control subsystems into a simulation function between the corresponding composite control
systems. Proposition~\ref{Prop:Composite_Control} then identifies this resulting simulation function,
together with its associated composite interface, with an elementary simulation morphism in
$\mathsf{CS}_{\mathrm{chain}}$.}

\final{Repeating this construction for every pair of adjacent model levels in the hierarchy produces a sequence of
elementary simulation morphisms. Since the target composite system of each morphism is the
source composite system of the next one, these morphisms can be concatenated into a chain
morphism across the network model hierarchy. The controller refinement and performance
transfer results in Theorem~\ref{Thm_Control_Refine} and
Corollary~\ref{Cor_CL_Performance_Transfer} can therefore be applied to the resulting chain.}

\begin{corollary}[Refinement across network model hierarchies]
\label{Cor:NetworkHierarchyRefinement}
Fir~$r \in \mathbb{N}_{> 0}$. 
For each $k \in \underline{r}$ and each $o\in\underline L$, let
\begin{align*}
\Delta_o^k
=
(\mathbb R^{n_o^k},\mathbb R^{m_o^k},\mathbb R^{p_o^k},
\mathcal U_o^k,\mathcal W_o^k,f_o^k,
\mathbb R^{q_o^k},\mathbb R^{s_o^k},h_o^k,y_o^k)
\end{align*}
denote the $o^{th}$ open control subsystem at the $k^{th}$ level of the model hierarchy.
Let $\bar{\boldsymbol{\Delta}}^k$ denote the corresponding composite control system obtained
by interconnecting the collection
$\{\Delta_o^k\}_{o=1}^L$
according to Definition~\ref{Def_CompositeSystem}.

\final{Suppose that, for each $k \in \underline{r} \backslash \{1\}$, 
the two subsystem collections
\begin{align*}
\{\Delta_o^{k-1}\}_{o=1}^L
\qquad\text{and}\qquad
\{\Delta_o^k\}_{o=1}^L
\end{align*}
satisfy the assumptions of Proposition~\ref{Thm:CompositionalSimulation}. 
For each~$k \in \underline{r}$, 
let~$\boldsymbol{V}^{k,k-1}$ denote the network-level simulation function
from~$\{\Delta_{o}^{k-1}\}_{o=1}^{L}$ to~$\{\Delta_o^k\}_{o=1}^{L}$, which is
accompanied by an interface selection rule
$\boldsymbol{\mathcal I}^{k,k-1}$ satisfying Definition~\ref{Def_Ele_Sim_Morph}.
Then the resulting
network-level simulation function and composite interface define an elementary simulation
morphism
\begin{align*}
\bar{\boldsymbol{\alpha}}^{k,k-1}
=
\bigl(
\boldsymbol V^{k,k-1},
\boldsymbol{\mathcal I}^{k,k-1}
\bigr):
\bar{\boldsymbol{\Delta}}^{k-1}
\longrightarrow
\bar{\boldsymbol{\Delta}}^k
\end{align*}
in $\mathsf{CS}_{\mathrm{chain}}$ for every $k \in \underline{r} \backslash \{1\}$.}

\final{Consequently, based on Definition~\ref{Def_chain_morphism_O}, these elementary simulation morphisms form the chain morphism}
\begin{align*}
\bar{\boldsymbol{\gamma}}
=
\bigl(
\bar{\boldsymbol{\alpha}}^{1,0},
\bar{\boldsymbol{\alpha}}^{2,1},
\ldots,
\bar{\boldsymbol{\alpha}}^{r,r-1}
\bigr):
\bar{\boldsymbol{\Delta}}^0
\longrightarrow
\bar{\boldsymbol{\Delta}}^r
\end{align*}
in $\mathsf{CS}_{\mathrm{chain}}$. Therefore, the controller refinement and performance transfer
conclusions of Theorem~\ref{Thm_Control_Refine} and
Corollary~\ref{Cor_CL_Performance_Transfer} hold for this chain of composite control systems. \hfill $\lozenge$
\end{corollary}

The proof is given in the Appendix. \final{Corollary~\ref{Cor:NetworkHierarchyRefinement} shows how subsystem composition and model
refinement interact. At each fixed model level, the open control
subsystems are interconnected through the corresponding composition matrix to form a composite
control system. Across adjacent model levels in the hierarchy, subsystem-level storage functions and the
interconnection compatibility conditions in Proposition~\ref{Thm:CompositionalSimulation} produce a simulation function between the corresponding
composite systems. Proposition~\ref{Prop:Composite_Control} expresses each such relation 
between composite systems as an
elementary simulation morphism in $\mathsf{CS}_{\mathrm{chain}}$.}

\final{These elementary simulation morphisms can then be concatenated across the model hierarchy.
Thus, the interconnection of subsystems within one model level and the categorical composition
of refinement relations between successive models in the hierarchy are two distinct operations. The former constructs
each composite control system, while the latter constructs a chain morphism between composite
systems of different levels of refinement in the hierarchy.}

\final{This construction provides a scalable route for controller refinement of large network models.
Rather than constructing a simulation function directly between full network models at every
pair of adjacent levels, one may construct storage functions between corresponding subsystems
and verify the associated interconnection conditions. The resulting network-level simulation
functions can then be organized into a chain morphism, allowing the controller refinement and
performance transfer results of Section~\ref{Sec_CTRL} to be applied across the entire network
model hierarchy.}

\section{Illustrative Examples of Controller Refinement}\label{Sec_Example_Simulation}
In this section, we present two thermal network examples to illustrate
recursive controller refinement, the output error bound in Theorem~\ref{Thm_Control_Refine},
and the compositional construction in Section~\ref{Sec:NetworkLifting}. The first example
considers a hierarchy of control systems without internal ports. The
second considers interconnected open control subsystems 
with internal ports
and shows how 
functions 
for these subsystems are combined at the network level.



\subsection{Multizone Thermal Model Hierarchy}
\label{Sec_sim1}
We illustrate the endpoint estimate in Theorem~\ref{Thm_Control_Refine} using a hierarchy of linearized multi-zone thermal
models without internal ports. Accordingly, the models
considered in this example are objects of $\mathsf{CS}_{\mathrm{chain}}$. The purpose
of this example is to show how a controller designed on a simple thermal model
can be transferred recursively to increasingly refined models and how the
differences between outputs of adjacent systems can be quantified. For each
adjacent pair, we select a simulation function 
and an input rule and evaluate the
corresponding conditions in Definition~\ref{Def: Simulation_Function} using the states and inputs generated
in the simulation.

We consider a simplified linearized multi-zone thermal network of resistance-capacitance (RC) type, motivated by
standard state space thermal models used in building control
\cite{atam2016control,abtahi2022control,yang2018state}. In
such models, zone temperatures are treated as states, while heat exchange between
neighboring zones, ambient losses, and HVAC inputs are represented in reduced state space
form~\cite{atam2016control,abtahi2022control,yang2018state}. 
The most refined model
represents a building with $12$ thermal zones arranged on a $3\times 4$ grid. Let
$x^4\in\mathbb R^{12}$ denote the vector of zone temperatures in the full model. The lower-dimensional models are obtained by successively aggregating zones
into larger clusters, which yields the hierarchy of five models
\begin{align*}
\bar\Delta^0
\longrightarrow
\bar\Delta^1
\longrightarrow
\bar\Delta^2
\longrightarrow
\bar\Delta^3
\longrightarrow
\bar\Delta^4,
\end{align*}
where $\bar\Delta^0$ is the simplest model and $\bar\Delta^4$ is the most
refined model. The corresponding state dimensions are
\begin{align*}
n^0=1,\qquad n^1=3,\qquad n^2=4,\qquad n^3=6,\qquad n^4=12.
\end{align*}

At each level $i\in\{0,1,2,3,4\}$, the model is 
\begin{align*}
\dot x^i = A^i x^i + B^i u^i + d^i,
\end{align*}
where $x^i\in\mathbb R^{n^i}$ is the vector of aggregate zone temperatures,
$u^i\in\mathbb R^{n^i}$ is the HVAC input, and $d^i\in\mathbb R^{n^i}$ is a constant term
capturing ambient temperature and internal heat gains. The matrix
$A^i\in\mathbb R^{n^i\times n^i}$ models both heat exchange between adjacent zones and heat
loss to the ambient environment, while $B^i\in\mathbb R^{n^i\times n^i}$ describes direct
HVAC actuation. In simulation, the reduced models are obtained by projecting the full
$12$-zone thermal network onto lower-dimensional clustered models.

To compare models of different dimensions, we embed their outputs in the common comparison
space $\mathbb R^{12}$. For each $i\in\{0,1,2,3,4\}$, let
$\Pi^i\in\mathbb R^{12\times n^i}$ denote the lifting matrix from the state space of
$\bar{\Delta}^i$ to~$\mathbb{R}^{12}$. We define the output map by
$h^i(x^i)=\Pi^i x^i$.
Thus, the outputs of all models have the same dimension and can be compared directly.

Let
$r_{\mathrm{ref}}(t)
=
22+0.8\sin(0.15t)$
be the scalar temperature reference, and define
$\boldsymbol r(t)
=
r_{\mathrm{ref}}(t)\mathbf 1_{12}$
in the common output space, where~$\mathbf{1}_{12}$ denotes the vector of all ones in~$\mathbb{R}^{12}$. 
The inputs of the more refined models are generated recursively
through the selected input rules. 
\textcolor{black}{In this example, each input rule maps the scalar control input of the preceding model to an input vector for the next model by repeating the same scalar value in all entries of the vector.} 
Thus, the controller is designed on the simplest model and transferred
from one model to the next.

For each $i\in\{0,1,2,3\}$, define
\begin{align*}
V^{i+1,i}(x^{i+1},x^i)
=
\left\|
\Pi^{i+1}x^{i+1}-\Pi^i x^i
\right\|^2.
\end{align*}
Since $h^i(x^i)=\Pi^i x^i$, we have that
\begin{align}
\label{Eq_simu_fun_Example1}
V^{i+1,i}(x^{i+1},x^i)
=
\left\|
h^{i+1}(x^{i+1})-h^i(x^i)
\right\|^2.
\end{align}
Therefore, the first condition in Definition~\ref{Def: Simulation_Function} holds with
$\lambda^{i+1,i}(s)=s^2$. Applying the triangle inequality to the differences 
of adjacent outputs 
gives
\begin{equation} \label{eq:bigpibound}
\left\|
\Pi^4x^4-\Pi^0x^0
\right\|
\le
\sum_{i=0}^{3}
\sqrt{V^{i+1,i}(x^{i+1},x^i)}.
\end{equation}

For the selected parameters, we also evaluate the dissipation inequality in
Definition~\ref{Def: Simulation_Function} using the states and inputs generated in the simulation. The
computed residuals are nonpositive over the simulated interval, \textcolor{black}{as shown in Figure~\ref{fig_Example_1_Diss_Residual}.}  This numerical
result supports the selected functions and input rules for this example.
\begin{figure}[!h]
    \centering
    \includegraphics[width=1\linewidth]{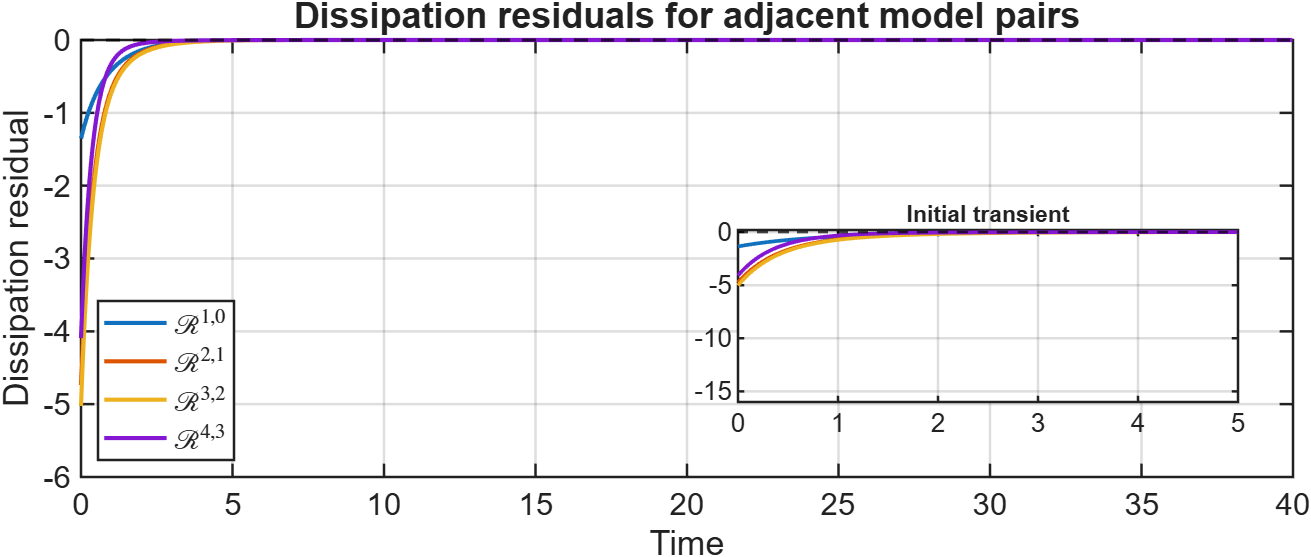}
    \caption{
   \textcolor{black}{Dissipation inequality residuals for the four adjacent model pairs in the model hierarchy. All residuals remain nonpositive over the simulated interval and approach zero as the adjacent model differences decay.}}
    \label{fig_Example_1_Diss_Residual}
\end{figure}

\begin{figure}[!h]
    \centering
    \includegraphics[width=1\linewidth]{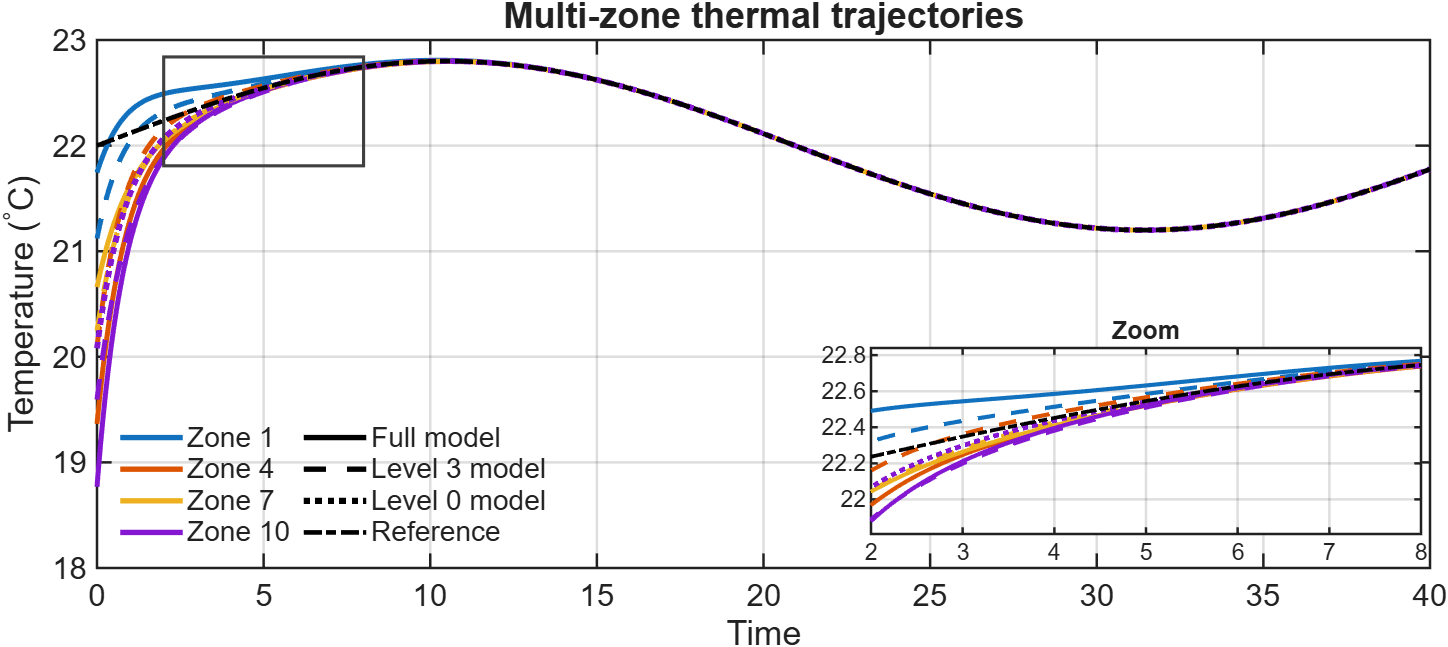}
    \caption{
    \textcolor{black}{Selected temperature trajectories of the full thermal model together with lifted trajectories of selected simpler models and the reference signal. Colors distinguish the selected zones, while line styles distinguish the different models in the hierarchy.}}
    \label{fig:Example_1_States}
\end{figure}
\begin{figure}
    \centering
    \includegraphics[width=1\linewidth]{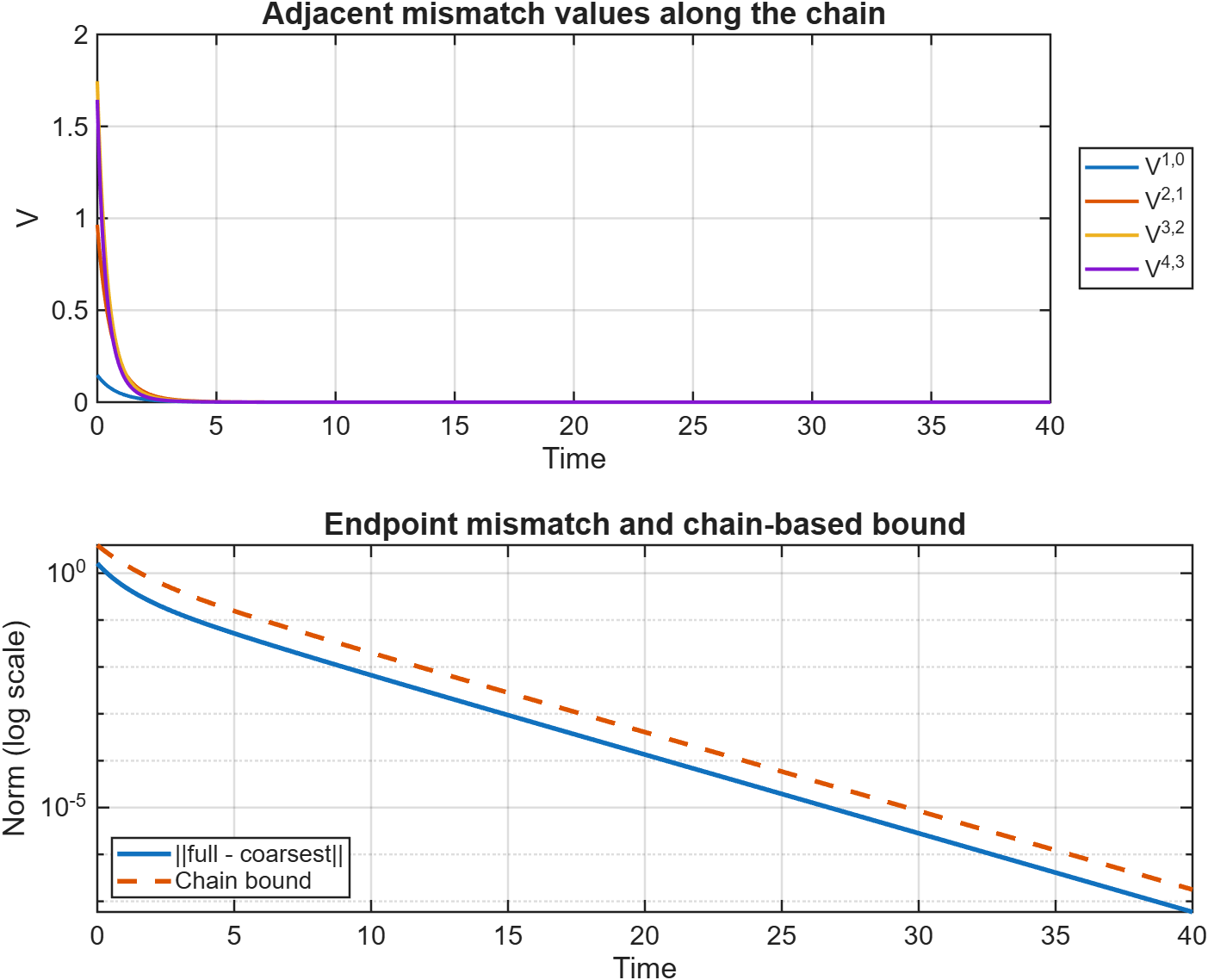}
    \caption{Mismatch quantities for the thermal model hierarchy. Top: adjacent
mismatch values. Bottom: endpoint mismatch and the sum of the adjacent
terms.}
    \label{fig:Example_1_Error}
\end{figure}
Figure~\ref{fig:Example_1_States} shows selected temperature trajectories of the
full thermal model together with the lifted trajectories of two simpler models and the
reference signal. The trajectories remain close over time, which illustrates that the
recursively generated inputs preserve the desired behavior between adjacent models in the hierarchy. In particular, the
full model and the lifted simpler models all track the same reference qualitatively, even
though the controller is designed only on the simplest model.
Figure~\ref{fig:Example_1_Error} shows the values of
$V^{1,0}$, $V^{2,1}$, $V^{3,2}$, and $V^{4,3}$, together with the endpoint mismatch
$\|\Pi^4x^4-\Pi^0x^0\|$ and the accumulated quantity
$\sum_{i=0}^{3}\sqrt{V^{i+1,i}(x^{i+1}, x^i)}$. The endpoint mismatch remains below the value of 
this quantity over
the simulated interval, as expected from~\eqref{eq:bigpibound} and the bound in Theorem~\ref{Thm_Control_Refine}.


To characterize the output mismatch after the initial transient, we evaluate the endpoint
mismatch for $t\in[10,40]$. Over this interval, we find 
\begin{align*}
\max_{t\in[10,40]}
\left\|
\Pi^4x^4(t)-\Pi^0x^0(t)
\right\|
=
6.6662\times10^{-3},
\end{align*}
while the maximum accumulated bound is
\begin{align*}
\max_{t\in[10,40]}
\sum_{i=0}^{3}
\sqrt{V^{i+1,i}(t)}
=
2.0075\times10^{-2}.
\end{align*}
Both maxima occur at $t=10$. The endpoint mismatch remains below the accumulated bound
throughout the interval~$[10, 40]$, 
and the maximum pointwise ratio of the mismatch to the bound is
\begin{equation}
\max_{t \in [10, 40]} \frac{\|\Pi^4x^4(t) - \Pi^0x^0(t)\|}{\sum_{i=0}^{3}\sqrt{V^{i+1,i}(x^{i+1}, x^i)}}
= 0.3333. 
\end{equation}
We also report the RMS temperature difference over the interval and across the zones, given by
\begin{align*}
E_{\mathrm{RMS}}
:=
\left(
\frac{1}{12(40-10)}
\int_{10}^{40}
\left\|
\Pi^4x^4(t)-\Pi^0x^0(t)
\right\|^2\,dt
\right)^{1/2}.
\end{align*}

Because the output vector contains the temperatures of $12$ zones, the maximum 
aggregate mismatch after the initial transient is
\begin{align*}
\frac{6.6662\times10^{-3}}{\sqrt{12}}
\approx
0.001924^\circ\mathrm{C}
\end{align*}
RMS per zone. The RMS temperature difference  over $t\in[10,40]$ and across the
$12$ zones is
$0.000397^\circ\mathrm{C}$. These results indicate that, after the initial
transient, the outputs of the simplest and most refined closed-loop models remain in close
agreement under the recursively refined input.

This example shows how a controller designed on the simplest one-state thermal model can be
transferred recursively to the most refined $12$-zone model. After the initial transient, the maximum
difference between the endpoint outputs is approximately
$0.001924^\circ\mathrm{C}$ RMS per zone, and the RMS difference over time is approximately
$0.000397^\circ\mathrm{C}$ RMS per zone. The endpoint difference remains below the sum of
the adjacent mismatch terms throughout the simulation. These results illustrate the controller
refinement and endpoint bound developed in Theorem~\ref{Thm_Control_Refine}.

\subsection{Three-Room Thermal Network at Three Levels of Refinement}
\label{Sec_sim2}
We illustrate the compositional construction of
Section~\ref{Sec:NetworkLifting} using a thermal system consisting of
three physically connected rooms. The rooms are indexed by
$o\in\{1,2,3\}$ and are arranged in a path, so that room~1 exchanges
heat with room~2, and room~2 exchanges heat with rooms~1 and~3.
Each physical room is represented at three levels of detail, indexed by
$i\in\{1,2,3\}$. For room $o$, the models
$\Delta_o^1$, $\Delta_o^2$, and $\Delta_o^3$ describe the same room
using increasingly refined thermal dynamics. Level~1 is the simplest
model, level~2 is an intermediate model, and level~3 is the most
refined model.

For each fixed level $i$, the three open control subsystems
$\Delta_1^i$, 
$\Delta_2^i$, and
$\Delta_3^i$
are interconnected through their internal inputs and outputs to form the
composite control system $\bar{\boldsymbol{\Delta}}^i$. Thus,
$\bar{\boldsymbol{\Delta}}^1$, $\bar{\boldsymbol{\Delta}}^2$, and $\bar{\boldsymbol{\Delta}}^3$ represent the same
three-room thermal network at three increasing levels of detail.
For each pair of adjacent levels, we define functions between the
corresponding models of each room. We then combine the three room
functions using the weighted sum construction in
Proposition~\ref{Thm:CompositionalSimulation} to obtain a function
between the corresponding composite systems.



For each room $o$, the simplest model $\Delta_o^1$ is a $1$R$1$C thermal model with state
$x_o^1=T_{a,o}^1\in\mathbb R$,
where $T_{a,o}^1$ denotes the room air temperature. Its dynamics are
\begin{align*}
C_{a,o}\dot T_{a,o}^1
=
-\frac{1}{R_{oa,o}}(T_{a,o}^1-T_{\mathrm{out}})
+u_o^1+w_o^1,
\end{align*}
and its external and internal outputs are both chosen as
$h_o^1(x_o^1)=T_{a,o}^1$ and
$y_o^1(x_o^1)=T_{a,o}^1$.
Thus, $n_o^1=1$. We use $C_{a,o}>0$ to denote the effective thermal capacitance associated with the room-air
state of room $o$, and $R_{oa,o}>0$ denotes the effective thermal resistance between the
room air and the outdoor environment. The term
$\frac{1}{R_{oa,o}}(T_{a,o}^1-T_{\mathrm{out}})$ represents conductive heat exchange with
the ambient environment, where~$T_{out}$ is a constant that denotes the temperature of the ambient environment. 
The input $u_o^1$ represents HVAC heat input, while $w_o^1$
represents the inter-room coupling signal induced by the network interconnection.
The intermediate model $\Delta_o^2$ is a $2$R$2$C thermal model with state $x_o^2
=
[T_{a,o}^2,T_{m,o}^2]^\top
\in\mathbb R^2$, 
where $T_{m,o}^2$ is an internal thermal mass temperature representing
heat stored in interior materials, furniture, partitions, and
equipment that respond more slowly than the room air. We use $C_{m,o}^2>0$ to denote the effective thermal capacitance associated with
the internal-mass state, and $R_{am,o}^2>0$ denotes the effective thermal resistance
between the room-air node and the internal-mass node. Its dynamics are
\begin{align*}
C_{a,o}\dot T_{a,o}^2
&=
-\frac{1}{R_{oa,o}}(T_{a,o}^2-T_{\mathrm{out}}) -\frac{1}{R_{am,o}^2}(T_{a,o}^2-T_{m,o}^2)
+u_o^2+w_o^2,\\
C_{m,o}^2\dot T_{m,o}^2
&=
\frac{1}{R_{am,o}^2}(T_{a,o}^2-T_{m,o}^2),
\end{align*}
where~$T_{out}$ is a constant that denotes the temperature of the ambient environment,
with outputs
$h_o^2(x_o^2)=T_{a,o}^2$
and 
$y_o^2(x_o^2)=T_{a,o}^2$.
Thus, $n_o^2=2$.

The finest model $\Delta_o^3$ is a $3$R$3$C thermal model with state $x_o^3= [T_{a,o}^3, T_{w,o}^3, T_{m,o}^3]^\top\in\mathbb R^3$, 
where $T_{a,o}^3$ is the room air temperature, $T_{w,o}^3$ is a lumped wall or envelope
temperature, and $T_{m,o}^3$ is a lumped internal thermal mass temperature representing
slower heat storage effects inside the room. The parameters $C_{w,o}^3>0$ and
$C_{m,o}^3>0$ denote the effective thermal capacitances associated with the wall and
internal thermal-mass states, respectively.
The parameter~$R_{aw,o}^3>0$ 
denotes the effective thermal resistance between the room-air and wall states,
and the parameter~$R_{wm,o}^3>0$
denotes the effective thermal resistances between 
the wall and internal thermal-mass states, respectively. Its dynamics are
\begin{align*}
C_{a,o}\dot T_{a,o}^3
&=
-\frac{1}{R_{oa,o}}(T_{a,o}^3-T_{\mathrm{out}}) -\frac{1}{R_{aw,o}^3}(T_{a,o}^3-T_{w,o}^3)
+u_o^3+w_o^3,\\
C_{w,o}^3\dot T_{w,o}^3
&=
\frac{1}{R_{aw,o}^3}(T_{a,o}^3-T_{w,o}^3)
-\frac{1}{R_{wm,o}^3}(T_{w,o}^3-T_{m,o}^3),\\
C_{m,o}^3\dot T_{m,o}^3
&=
\frac{1}{R_{wm,o}^3}(T_{w,o}^3-T_{m,o}^3),
\end{align*}
where~$T_{out}$ is a constant that denotes the temperature of the ambient environment, 
with outputs
$h_o^3(x_o^3)=T_{a,o}^3$
and
$y_o^3(x_o^3)=T_{a,o}^3$.
Thus, $n_o^3=3$.

For each level $i\in\{1,2,3\}$, define the state, input, internal-input,
and external-output vectors of the composite system by
\begin{align*}
\boldsymbol x^i
&:=
\begin{bmatrix}
(x_1^i)^\top &
(x_2^i)^\top &
(x_3^i)^\top
\end{bmatrix}^{\top},
\ \ \boldsymbol u^i
:=
\begin{bmatrix}
u_1^i & u_2^i & u_3^i
\end{bmatrix}^{\top},\\
\boldsymbol w^i
&:=
\begin{bmatrix}
w_1^i & w_2^i & w_3^i
\end{bmatrix}^{\top}, 
\ \
\boldsymbol h^i(\boldsymbol x^i)
:=
\begin{bmatrix}
h_1^i(x_1^i) &
h_2^i(x_2^i) &
h_3^i(x_3^i)
\end{bmatrix}^{\top}.
\end{align*}
We also define the stacked temperature vectors
\begin{align*}
\boldsymbol T_a^i
&:=
\begin{bmatrix}
T_{a,1}^i &
T_{a,2}^i &
T_{a,3}^i
\end{bmatrix}^{\top},
\ \ i\in\{1,2,3\}, \quad
\boldsymbol T_m^2
:=
\begin{bmatrix}
T_{m,1}^2 &
T_{m,2}^2 &
T_{m,3}^2
\end{bmatrix}^{\top},\\
\boldsymbol T_w^3
&:=
\begin{bmatrix}
T_{w,1}^3 &
T_{w,2}^3 &
T_{w,3}^3
\end{bmatrix}^{\top}, \quad
\boldsymbol T_m^3
:=
\begin{bmatrix}
T_{m,1}^3 &
T_{m,2}^3 &
T_{m,3}^3
\end{bmatrix}^{\top}.
\end{align*}
Since the external output of each room is its room air temperature, we have 
$\boldsymbol h^i(\boldsymbol x^i)
=
\boldsymbol T_a^i$.
We collect the room-dependent resistance parameters in the vectors
\begin{align*}
R_{oa}
&:=
\begin{bmatrix}
R_{oa,1}&R_{oa,2}&R_{oa,3}
\end{bmatrix}^{\top}, \quad
R_{am}^{2}
:=
\begin{bmatrix}
R_{am,1}^{2}&R_{am,2}^{2}&R_{am,3}^{2}
\end{bmatrix}^{\top},\\
R_{aw}^{3}
&:=
\begin{bmatrix}
R_{aw,1}^{3}&R_{aw,2}^{3}&R_{aw,3}^{3}
\end{bmatrix}^{\top}, \quad
R_{wm}^{3}
:=
\begin{bmatrix}
R_{wm,1}^{3}&R_{wm,2}^{3}&R_{wm,3}^{3}
\end{bmatrix}^{\top}.
\end{align*}
Division operations involving these vectors are interpreted componentwise. We also
define
\begin{align*}
\mathsf{C}_a
&:=
\operatorname{diag}
\left(C_{a,1},C_{a,2},C_{a,3}\right), \quad
\mathsf{D}_{oa}
:=
\operatorname{diag}
\left(
R_{oa,1}^{-1},
R_{oa,2}^{-1},
R_{oa,3}^{-1}
\right).
\end{align*}
Both $\mathsf{C}_a$ and $\mathsf{D}_{oa}$ are positive definite.

For each fixed model level $i$, the three room subsystems are interconnected according to
$\boldsymbol w^i=M\boldsymbol y^i$,
where
\begin{align*}
\boldsymbol y^i
:=
\begin{bmatrix}
y_1^i(x_1^i) &
y_2^i(x_2^i) &
y_3^i(x_3^i)
\end{bmatrix}^{\top},
\end{align*} 
and $M=-\kappa L$ with $\kappa>0$.
Note that $L$
is the Laplacian matrix of the three-room path graph. Thus, each internal input is determined by
the differences between the room air temperature of that room and the room-air temperature(s) of its neighbor(s). The
same interconnection structure is used at all three model levels in the hierarchy, and the
resulting interconnected systems are denoted by
$\bar{\boldsymbol{\Delta}}^1$,
$\bar{\boldsymbol{\Delta}}^2$, and 
$\bar{\boldsymbol{\Delta}}^3$.
Equivalently, the internal inputs are
\begin{align*}
w_1^i
=
\kappa\left(T_{a,2}^i-T_{a,1}^i\right), \quad
w_2^i
=
\kappa\left(
T_{a,1}^i-2T_{a,2}^i+T_{a,3}^i
\right), \quad
w_3^i 
=
\kappa\left(T_{a,2}^i-T_{a,3}^i\right).
\end{align*}
Thus, rooms~1 and~3 exchange heat only with room~2, while room~2
exchanges heat with both room~1 and room~3.

The control law is designed only on the coarsest interconnected model
$\bar{\boldsymbol{\Delta}}^1$. Let
$T_{\mathrm{ref}}(t)=22+0.5\sin(0.20t)$
be the desired reference room temperature. We define the stacked reference signal and its
derivative by
\begin{align*}
\boldsymbol {T}_r(t)&=T_{\mathrm{ref}}(t)\mathbf 1_3, \qquad
\dot {\boldsymbol {T}}_r(t)=0.10\cos(0.20t)\mathbf 1_3.
\end{align*}
The level-$1$ input is chosen as
\begin{equation}
\boldsymbol u^1(t)
=
\mathsf C_a\dot{\boldsymbol T}_r(t)
+
\mathsf D_{oa}
\bigl(
\boldsymbol T_r(t)-T_{\mathrm{out}}\mathbf 1_3
\bigr)
-
M\boldsymbol T_r(t)
-
K_{\mathrm{tr}}
\bigl(
\boldsymbol T_a^1(t)-\boldsymbol T_r(t)
\bigr).
\end{equation}
This law is obtained by
matching the nominal level-$1$ thermal dynamics to the desired reference evolution and
adding proportional feedback on the room air temperature tracking error. In particular,
$\mathsf C_a\dot{\boldsymbol T}_r(t)$ accounts for the room-air
capacitances in the desired reference dynamics,
$\mathsf D_{oa}\bigl(\boldsymbol {T}_r(t)-T_{\mathrm{out}}\mathbf 1_3\bigr)$
compensates for ambient heat loss at the reference, $-M\boldsymbol {T}_r(t)$ compensates for the diffusive interconnection at the reference,
and the last term stabilizes the tracking error, where $K_{\mathrm{tr}}>0$ is the tracking gain.

The inputs for levels $2$ and $3$ are then generated recursively through adjacent
refinement laws, given by
\begin{align*}
\boldsymbol{u}^2(t)
&=
\boldsymbol{u}^1(t)
+\frac{\boldsymbol{T}_a^2(t)-\boldsymbol{T}_m^2(t)}{R_{am}^2}
-K_{21}\big(\boldsymbol{T}_a^2(t)-\boldsymbol{T}_a^1(t)\big),\\
\boldsymbol{u}^3(t)
&=
\boldsymbol{u}^2(t)
+\frac{\boldsymbol{T}_a^3(t)-\boldsymbol{T}_w^3(t)}{R_{aw}^3}
-K_{32}\big(\boldsymbol{T}_a^3(t)-\boldsymbol{T}_a^2(t)\big).
\end{align*}
The additional correction terms compensate for the hidden
thermal couplings that appear in the more refined models, while the gains $K_{21}>0$
and $K_{32}>0$ penalize the adjacent visible output mismatch. The resulting adjacent room air temperature error dynamics are analyzed next.



Define the adjacent room air temperature errors by
\begin{align*}
\boldsymbol{e}^{2,1}
:=\boldsymbol{T}_a^2-\boldsymbol{T}_a^1, \quad
\boldsymbol{e}^{3,2}
:=\boldsymbol{T}_a^3-\boldsymbol{T}_a^2.
\end{align*}
Subtracting the adjacent model dynamics under the selected input refinement
laws gives
\begin{align*}
\mathsf C_a\dot{\boldsymbol{e}}^{2,1}
&=
-\left(\mathsf D_{oa}+K_{21}I_3\right)\boldsymbol{e}^{2,1}
+\left(\boldsymbol{w}^2-\boldsymbol{w}^1\right),\\
\mathsf C_a\dot{\boldsymbol{e}}^{3,2}
&=
-\left(\mathsf D_{oa}+K_{32}I_3\right)\boldsymbol{e}^{3,2}
+\left(\boldsymbol{w}^3-\boldsymbol{w}^2\right).
\end{align*}
Because the same diffusive interconnection matrix is used at every model
level and
$\boldsymbol h^i(\boldsymbol x^i)=\boldsymbol T_a^i$, we have
\begin{align*}
\boldsymbol w^2-\boldsymbol w^1
&=
M\left(
\boldsymbol h^2(\boldsymbol x^2)
-
\boldsymbol h^1(\boldsymbol x^1)
\right)
=
M\boldsymbol e^{2,1},\\
\boldsymbol w^3-\boldsymbol w^2
&=
M\left(
\boldsymbol h^3(\boldsymbol x^3)
-
\boldsymbol h^2(\boldsymbol x^2)
\right)
=
M\boldsymbol e^{3,2}.
\end{align*}
Therefore,
\begin{align}
\mathsf C_a\dot {\boldsymbol{e}}^{2,1}
&=
-\left(\mathsf D_{oa}+K_{21}I_3-M\right){\boldsymbol{e}}^{2,1},
\label{eq:error_dynamics_21}\\
\mathsf C_a\dot {\boldsymbol{e}}^{3,2}
&=
-\left(\mathsf D_{oa}+K_{32}I_3-M\right){\boldsymbol{e}}^{3,2}.
\label{eq:error_dynamics_32}
\end{align}

Since $M=-\kappa L$, where $L\succeq0$, the matrices
\begin{align*}
Q_{21}
&:=
\mathsf D_{oa}+K_{21}I_3-M
=
\mathsf D_{oa}+K_{21}I_3+\kappa L,\\
Q_{32}
&:=
\mathsf D_{oa}+K_{32}I_3-M
=
\mathsf D_{oa}+K_{32}I_3+\kappa L
\end{align*}
are positive definite for $K_{21}>0$ and $K_{32}>0$. For example, consider
$\mathcal E_{21}
=
\frac{1}{2}(\boldsymbol{e}^{2,1})^\top
\mathsf C_a \boldsymbol{e}^{2,1}$.
Along trajectories produced by~\eqref{eq:error_dynamics_21}, we have
\begin{align*}
\dot{\mathcal E}_{21}
=
-(\boldsymbol e^{2,1})^\top
Q_{21}\boldsymbol e^{2,1}
\le
-\lambda_{\min}(Q_{21})
\|\boldsymbol e^{2,1}\|^2
\le
-\frac{
2\lambda_{\min}(Q_{21})
}{
\lambda_{\max}(\mathsf C_a)
}
\mathcal E_{21}.
\end{align*}
Therefore, $\boldsymbol e^{2,1}$ converges to zero exponentially.
The same argument applies to $\boldsymbol e^{3,2}$ using $Q_{32}$. Hence, both adjacent room air temperature error
systems are exponentially stable.

To quantify the differences between adjacent subsystem models, we define the 
following functions at the subsystem level: 
\begin{align*}
V_o^{2,1}(x_o^2,x_o^1)
&=
(T_{a,o}^2-T_{a,o}^1)^2
+\alpha_o^2(T_{m,o}^2-T_{a,o}^2)^2,\\
V_o^{3,2}(x_o^3,x_o^2)
&=
(T_{a,o}^3-T_{a,o}^2)^2
+\alpha_{w,o}^3(T_{w,o}^3-T_{a,o}^3)^2 
+\alpha_{m,o}^3(T_{m,o}^3-T_{w,o}^3)^2.
\end{align*}
The first term in each function is the squared difference between the adjacent room-air
temperatures. The remaining terms account for differences involving the additional thermal
states in the more refined models. The coefficients
$\alpha_o^2$, $\alpha_{w,o}^3$, and $\alpha_{m,o}^3$
are positive weights assigned to the additional thermal-state terms.
In the simulations, we use
$\alpha_o^2
=
\alpha_{w,o}^3
=
\alpha_{m,o}^3
=
0.8, \textnormal{ for } 
o\in\{1,2,3\}$.


In this example, we choose equal weights
$\mu_o^{2,1}=\mu_o^{3,2}=1, \textnormal{ for } o\in\{1,2,3\}$.
The corresponding functions for the interconnected system are therefore
\begin{align*}
\boldsymbol V^{2,1}(\boldsymbol x^2,\boldsymbol x^1)
&=
\sum_{o=1}^{3}
V_o^{2,1}(x_o^2,x_o^1), \quad
\boldsymbol V^{3,2}(\boldsymbol x^3,\boldsymbol x^2)
=
\sum_{o=1}^{3}
V_o^{3,2}(x_o^3,x_o^2).
\end{align*}
Proposition~\ref{Thm:CompositionalSimulation} states the conditions under which these weighted sums are simulation functions
between the composite systems. We evaluate these conditions numerically for the selected
parameters and simulations.

Since the external outputs are the stacked room air temperatures, the endpoint output mismatch
between the coarsest and most refined composite systems is
\begin{align*}
\left\|
\boldsymbol h^3(\boldsymbol x^3)
-
\boldsymbol h^1(\boldsymbol x^1)
\right\|
=
\left\|\boldsymbol T_a^3-\boldsymbol T_a^1\right\|.
\end{align*}
For the choices
$\lambda^{2,1}(s)=\lambda^{3,2}(s)=s^2,$
the endpoint estimate in Theorem~\ref{Thm_Control_Refine} gives
\begin{align*}
\left\|\boldsymbol T_a^3-\boldsymbol T_a^1\right\|
\le
\sqrt{\boldsymbol V^{2,1}(\boldsymbol x^2,\boldsymbol x^1)}
+
\sqrt{\boldsymbol V^{3,2}(\boldsymbol x^3,\boldsymbol x^2)}.
\end{align*}

For the selected parameters, we evaluate the subsystem dissipation inequalities using the states
and inputs generated in the simulations. We use
$\rho_o^{j,i}(s)=\gamma_o^{j,i}s^2$. 
For each room and adjacent model pair, $\gamma_o^{j,i}$ is selected
numerically so that the evaluated residual is nonpositive. The maximum
residual is zero to the reported numerical precision. We list these parameters in Table~\ref{table_1}.
\begin{table}[h]
\label{table_1}
\centering
\caption{Parameters used in the evaluated subsystem dissipation
inequalities and the resulting maximum residuals.}
\label{tab:sim2_residuals}
\begin{tabular}{c c c c}
\hline
Room & Pair & $\gamma_o^{j,i}$ & Maximum residual\\
\hline
$1$ & $(2,1)$ & $1.7705\times 10^{-2}$ & $0$\\
$1$ & $(3,2)$ & $5.6881\times 10^{-3}$ & $0$\\
$2$ & $(2,1)$ & $9.8674\times 10^{-4}$ & $0$\\
$2$ & $(3,2)$ & $1.8324\times 10^{-3}$ & $0$\\
$3$ & $(2,1)$ & $1.5806\times 10^{-2}$ & $0$\\
$3$ & $(3,2)$ & $3.7512\times 10^{-3}$ & $0$\\
\hline
\end{tabular}
\end{table}
We also evaluate the matrix inequality and the matching condition in Proposition~\ref{Thm:CompositionalSimulation} for the
pairs $(2,1)$ and $(3,2)$. The computed conditions are satisfied for the selected parameters.
These calculations support the use of the weighted functions in this example, but they do not
show that all the requirements of Definition~\ref{Def: Storage_Function} hold for every possible state and input.


\begin{figure}
    \centering
    \includegraphics[width=1\linewidth]{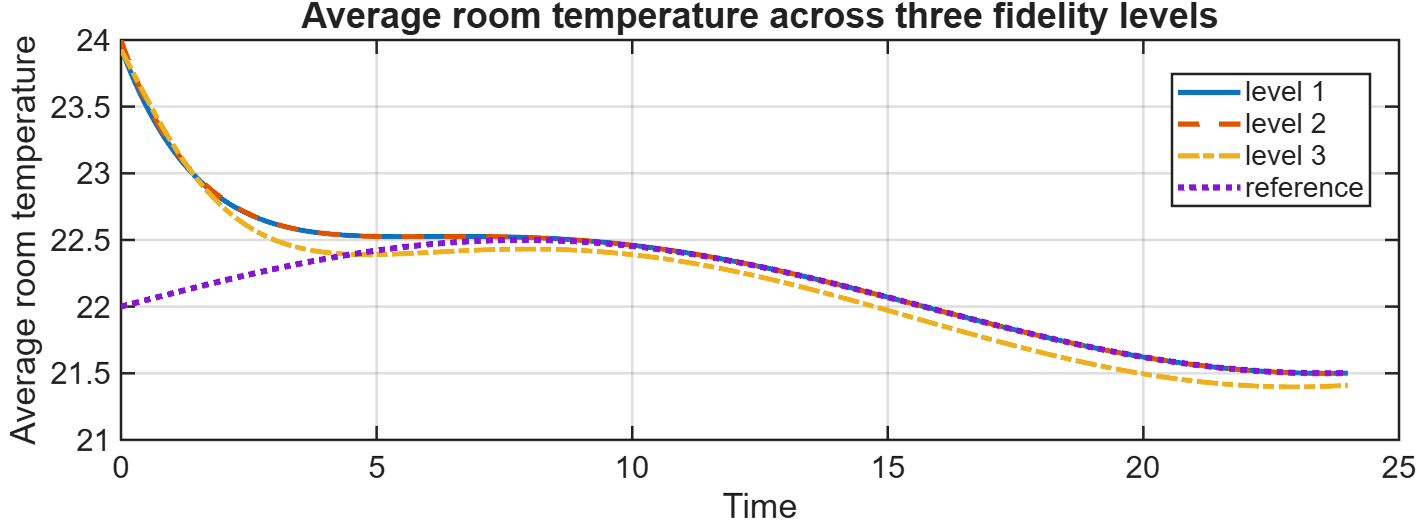}
    \caption{Average room air temperatures at the three model levels in the hierarchy together with the
    reference profile.}
    \label{fig:thermal_temperature}
\end{figure}

\begin{figure}
    \centering
    \includegraphics[width=1\linewidth]{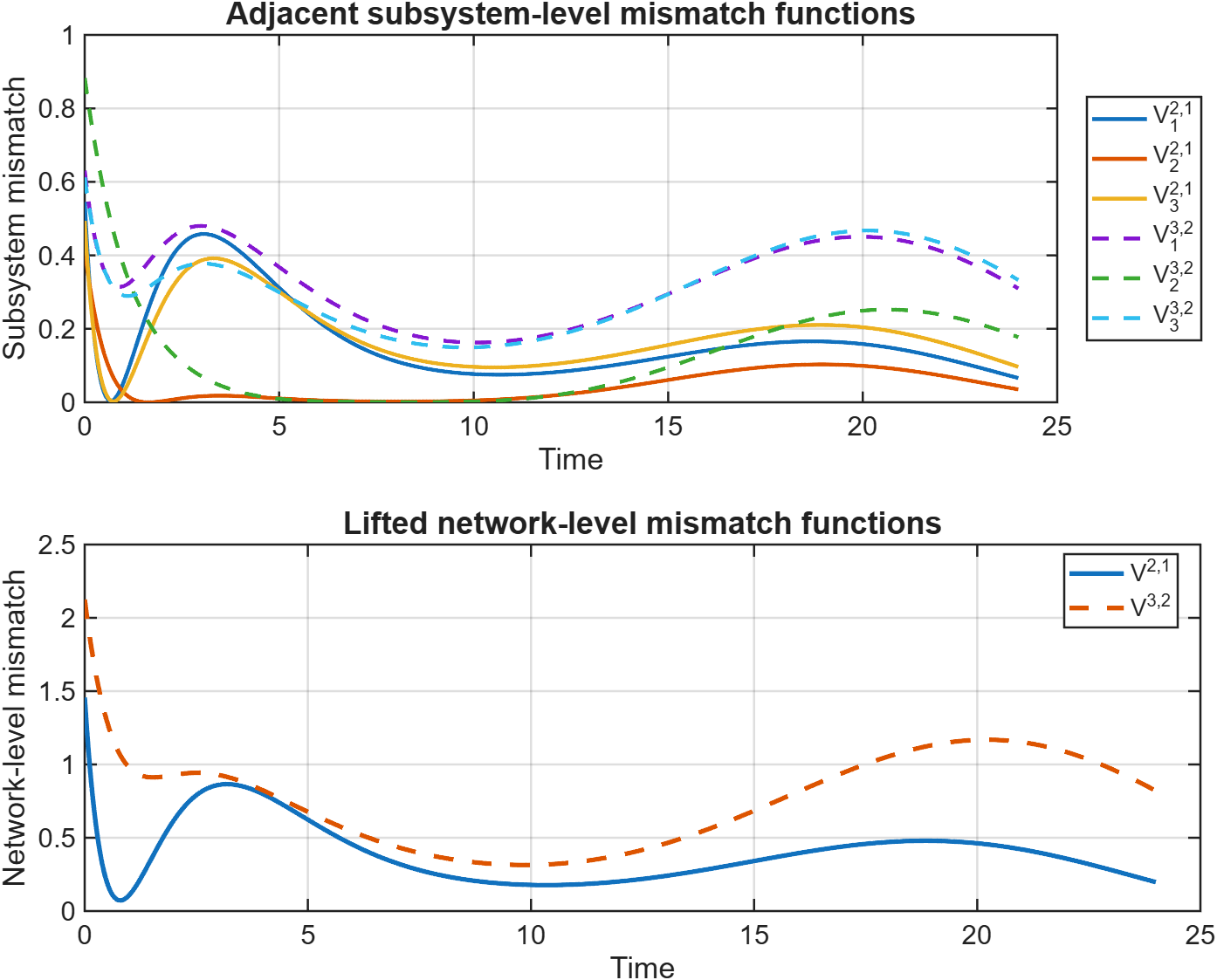}
    \caption{Mismatch functions for the individual rooms and their weighted
sums for the interconnected systems. Top: functions for the three rooms
for the adjacent pairs $(2,1)$ and $(3,2)$. Bottom:
$\boldsymbol V^{2,1}$ and $\boldsymbol V^{3,2}$.}
    \label{fig:thermal_storage}
\end{figure}

\begin{figure}
    \centering
    \includegraphics[width=1\linewidth]{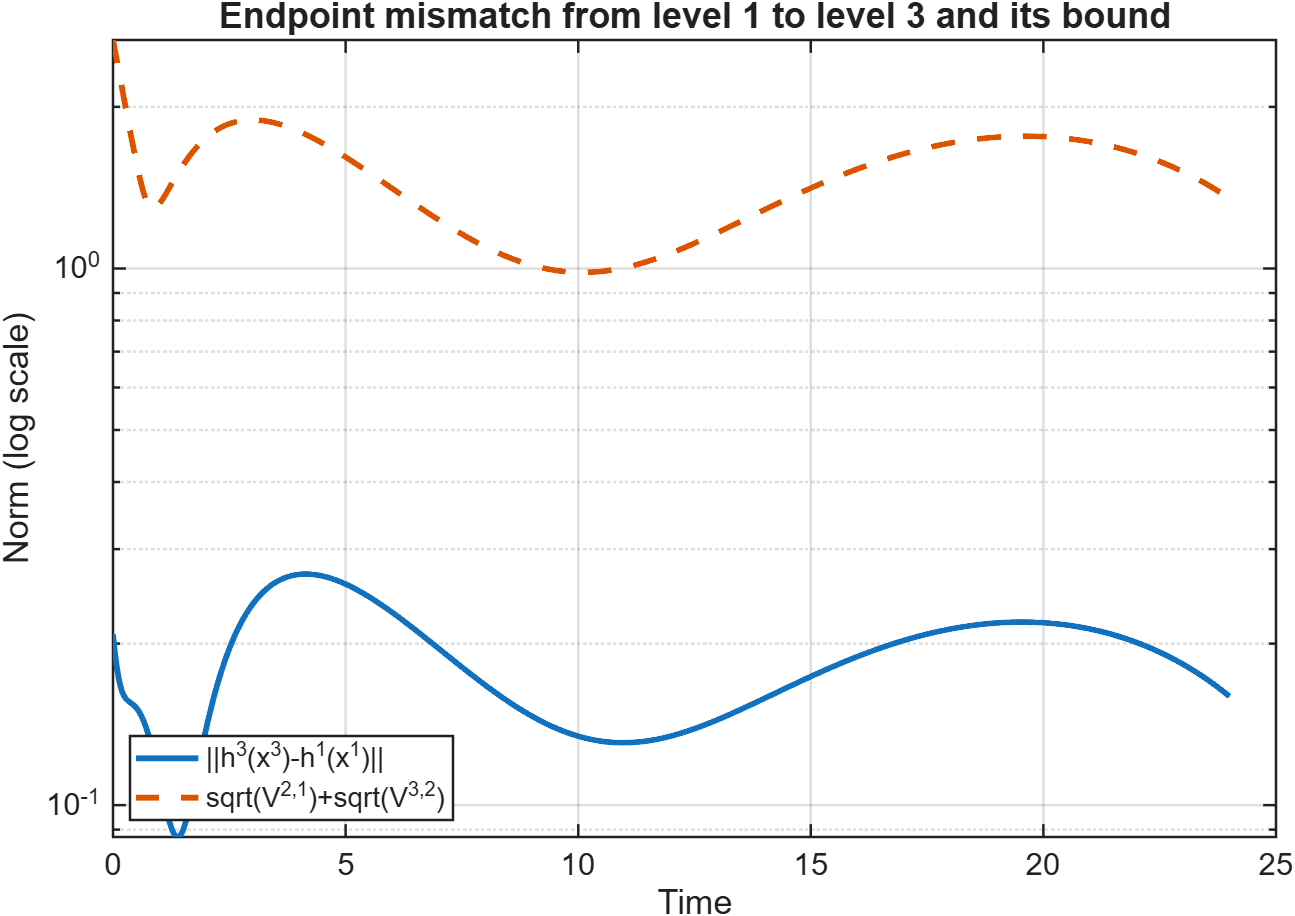}
    \caption{Endpoint output mismatch between the most refined and simplest interconnected
    models and the sum of the adjacent terms.}
    \label{fig:thermal_endpoint}
\end{figure}
Figure~\ref{fig:thermal_temperature} shows the average room air temperatures at the three
model levels in the hierarchy together with the reference profile. The trajectories illustrate that the
controller designed for the level-$1$ composite model can be transferred recursively to the
more refined composite models.
Figure~\ref{fig:thermal_storage} shows the functions for the individual
rooms and their weighted sums for the interconnected systems. These plots
illustrate how the three room functions are combined to form
$\boldsymbol V^{2,1}$ and $\boldsymbol V^{3,2}$.
Figure~\ref{fig:thermal_endpoint} compares the endpoint mismatch
$\left\|
\boldsymbol T_a^3-\boldsymbol T_a^1
\right\|$
with the sum of the adjacent terms
$
\sqrt{\boldsymbol V^{2,1}}
+
\sqrt{\boldsymbol V^{3,2}}$.
The endpoint mismatch remains below this sum throughout the simulated
trajectory, in agreement with the endpoint estimate in
Theorem~\ref{Thm_Control_Refine}.

This example illustrates both forms of composition considered in the
framework. Within each model level, the three room subsystems are
interconnected through the matrix $M$ to form a composite control system.
Across model levels in the hierarchy, the functions associated with adjacent composite systems
are combined as in the endpoint estimate of Theorem~\ref{Thm_Control_Refine}. Under the assumptions
of Proposition~\ref{Thm:CompositionalSimulation} and Definition~\ref{Def: Simulation_Function}, the corresponding relations define
elementary simulation morphisms and can be concatenated into a chain morphism
in $\mathsf{CS}_{\mathrm{chain}}$. In this example, the relevant conditions are
evaluated numerically for the selected parameters and simulation.

To characterize the mismatch after the initial transient, we evaluate the
endpoint mismatch over $t\in[10,24]$. Over this interval,
\begin{align*}
\max_{t\in[10,24]}
\left\|
T_a^3(t)-T_a^1(t)
\right\|
=
0.3067,
\end{align*}
and the maximum occurs at $t=19.79$.
The maximum value of the corresponding sum of adjacent terms is
\begin{align*}
\max_{t\in[10,24]}
\left(
\sqrt{\boldsymbol V^{2,1}(t)}
+
\sqrt{\boldsymbol V^{3,2}(t)}
\right)
=
1.5034,
\end{align*}
which occurs at $t=19.57$. The maximum ratio between the endpoint
difference and this sum over~$[10,24]$ is $0.2284$.

Since the output vector contains the room air temperatures of three
rooms, the maximum aggregate mismatch
\[
\max_{t\in[10,24]}
\left\|
\boldsymbol T_a^3(t)-\boldsymbol T_a^1(t)
\right\|
=
0.3067
\]
corresponds to a maximum RMS temperature difference across the three
rooms of
\begin{align*}
\frac{0.3067}{\sqrt{3}}
\approx
0.1771^\circ\mathrm C.
\end{align*}
We also report the RMS temperature difference over the interval and
across the three rooms, given by
\begin{align*}
E_{\mathrm{RMS}}
:=
\left(
\frac{1}{3(24-10)}
\int_{10}^{24}
\left\|
\boldsymbol T_a^3(t)-\boldsymbol T_a^1(t)
\right\|^2\,dt
\right)^{1/2}.
\end{align*}
For this simulation,
$E_{\mathrm{RMS}}=0.1482^\circ\mathrm C$. Therefore, after the initial transient, the controller
designed on the simplest interconnected model and recursively refined to the most refined model
produces room air temperature outputs that differ by at most approximately
$0.1771^\circ\mathrm{C}$ RMS per room along the simulated trajectory.

This example shows how the functions for the individual rooms can be combined across the three subsystems
and then used to compare the interconnected models at different levels of detail. Over the
interval after the initial transient, the maximum endpoint output difference is approximately
$0.1771^\circ\mathrm{C}$ RMS per room, and the RMS temperature difference over the interval and across the three
rooms is approximately
$0.1482^\circ\mathrm{C}$ RMS per room. The endpoint difference remains below the sum of
the two adjacent terms throughout this interval. 

\section{Conclusion and Future Work}
\label{Sec_Conclusion}

We developed a category-theoretic framework for controller refinement across hierarchies of
models of control systems. This framework represents a refinement relation between adjacent models
as an elementary morphism consisting of a storage function and an interface. Sequences of
composable elementary morphisms form chain morphisms, which allow a controller designed for a
simpler model to be constructed recursively for a more refined model even when a direct relation
between the source and target systems along the chain are unavailable.
The main results provide bounds on the endpoint output mismatch and on the transfer of
closed-loop performance between models. These bounds are obtained by accumulating the terms from the storage functions
associated with adjacent model pairs along the chain. Thus, the framework separates a
potentially difficult endpoint comparison into a sequence of local comparisons between
neighboring model levels in the hierarchy.

We also considered control systems without internal ports. We introduced the category
$\mathsf{CS}_{\mathrm{chain}}$ and defined a functor that projects from
$\mathsf{OCS}_{\mathrm{chain}}$ to $\mathsf{CS}_{\mathrm{chain}}$ that removes the structure of the internal ports while
preserving identities, composition, and chain structure. Consequently, the controller refinement
and performance transfer results remain available after projection.
For interconnected systems, we recalled a compositional result that combines storage functions
between corresponding open control subsystems into a simulation function between the
corresponding composite control systems under interconnection compatibility conditions.
We then expressed this result in category-theoretic terms by showing that the resulting
simulation function and composite interface define an elementary simulation morphism in
$\mathsf{CS}_{\mathrm{chain}}$. Repeating this construction across adjacent model levels in the hierarchy produces a chain
morphism between composite network models.

Thermal network examples illustrated the two main settings. The
first showed how a controller can be transferred through a hierarchy
of systems without internal ports. The second showed how functions for
individual subsystems can be combined for an interconnected system.
After the initial transient, the maximum endpoint output differences
were approximately $0.001924^\circ\mathrm C$ RMS per zone and
$0.1771^\circ\mathrm C$ RMS per room, respectively. In both examples,
the endpoint difference remained below the corresponding sum of the
adjacent terms.

Future work will investigate conditions that reduce conservatism in the 
bounds along the chain and extend the framework to stochastic, hybrid, uncertain, and data-driven systems. Another research direction is to consider more general interconnection structures
and to investigate how categorical constructions can support modular controller design across
models with different levels of refinement.

\section*{Appendix}
\subsection*{Proof of Theorem~\ref{Thm_CTRL_Path}}
\begin{proof}
Recall from Definition~\ref{Def:Category} that, in order to specify a category, we need to define a collection of objects, 
an identity morphism for each object, 
a set of morphisms between every ordered pair of objects, and a composition law for compatible morphisms.
Then we need to verify the unitality and associativity properties. We show that $\mathsf{OCS}_{\mathrm{chain}}$ satisfies each of these requirements.

First, by construction, the objects of $\mathsf{OCS}_{\mathrm{chain}}$ are all open control systems in $\mathsf{Sys}$, where $\mathsf{Sys}$ is the collection of open control systems in Definition~\ref{Def_Con_Systm}. Hence, the objects are well defined.
We next show the existence of an identity morphism for every object in $\mathsf{Sys}$. Let $\Delta\in\mathsf{Sys}$ be arbitrary. By definition, we set
$\mathrm{id}_{\Delta}:=()$,   
that is, the empty chain from $\Delta$ to itself. This empty chain is a valid chain morphism from $\Delta$ to itself. It has source object $\Delta$, target object $\Delta$, and contains no incompatible elementary morphisms. Therefore, every object $\Delta\in\mathsf{Sys}$ admits an identity morphism $\mathrm{id}_{\Delta}$ in $\mathsf{OCS}_{\mathrm{chain}}$.

We next show that, for every ordered pair of systems $\Delta^0,\Delta^r\in\mathsf{Sys}$, the hom-set $\mathsf{OCS}_{\mathrm{chain}}(\Delta^0,\Delta^r)$ is well defined. By Definition~\ref{Def_chain_morphism}, a morphism $\gamma:\Delta^0\to\Delta^r$ is a finite sequence of composable elementary morphisms
\begin{align*}
\gamma=(\alpha^{1,0},\alpha^{2,1},\ldots,\alpha^{r,r-1}),  
\end{align*}
where each $\alpha^{k+1,k}:\Delta^k\to\Delta^{k+1}$ is an elementary morphism given by Definition~\ref{Def_Ele_stor_morph}. \final{Therefore, for any pair $\Delta^0,\Delta^r\in\mathsf{Sys}$, we define the hom-set
$\mathsf{OCS}_{\mathrm{chain}}(\Delta^0,\Delta^r)$ to be the set of all chain morphisms from
$\Delta^0$ to $\Delta^r$. This set may be empty when no chain of elementary morphisms connects
$\Delta^0$ to $\Delta^r$. However, when $\Delta^0=\Delta^r$, the hom-set contains at least the
identity morphism~$\mathrm{id}_{\Delta}$, namely the empty chain. Hence, the hom-sets in $\mathsf{OCS}_{\mathrm{chain}}$ are well defined.}

We next verify that composition of morphisms is well defined. Consider two morphisms
\begin{align*}
\gamma:\Delta^0\to\Delta^r,\qquad
\gamma':\Delta^r\to\Delta^s,    
\end{align*}
where
\begin{align*}
\gamma=(\alpha^{1,0},\alpha^{2,1},\ldots,\alpha^{r,r-1})    
\end{align*}
and
\begin{align*}
\gamma'=(\alpha'^{r+1,r},\alpha'^{r+2,r+1},\ldots,\alpha'^{s,s-1}).    
\end{align*}
Since the target object of $\gamma$ is the same as the source object of $\gamma'$, namely $\Delta^r$, we can concatenate the two chains and define
\begin{equation} \label{eq:gammacomp}
\gamma'\circ\gamma
:=
(\alpha^{1,0},\alpha^{2,1},\ldots,\alpha^{r,r-1},
\alpha'^{r+1,r},\alpha'^{r+2,r+1},\ldots,\alpha'^{s,s-1}).
\end{equation}
This concatenated sequence is again a finite sequence of composable elementary morphisms. Indeed, the elementary morphisms within $\gamma$ are composable by Definition~\ref{Def_chain_morphism}, the elementary morphisms within $\gamma'$ are also composable by Definition~\ref{Def_chain_morphism},
and the target object of $\gamma$ agrees with the source object of $\gamma'$. Hence, $\gamma'\circ\gamma$ is a chain morphism from $\Delta^0$ to $\Delta^s$. Therefore, composition is well defined.

We next verify the unitality property. Let
\begin{align*}
\gamma=(\alpha^{1,0},\alpha^{2,1},\ldots,\alpha^{r,r-1}):\Delta^0\to\Delta^r
\end{align*}
be any morphism in $\mathsf{OCS}_{\mathrm{chain}}$. By the definitions of the identity morphism and 
the aforementioned
composition law, we have \begin{align*}
\gamma\circ \mathrm{id}_{\Delta^0}
=
(\alpha^{1,0},\alpha^{2,1},\ldots,\alpha^{r,r-1})
=
\gamma
\end{align*}
since concatenating the empty chain on the left changes nothing. Similarly,
\begin{align*}
\mathrm{id}_{\Delta^r}\circ\gamma
=
(\alpha^{1,0},\alpha^{2,1},\ldots,\alpha^{r,r-1})
=
\gamma,
\end{align*}
since concatenating the empty chain on the right also changes nothing. Thus, the identity morphisms satisfy the left and right identity laws.

Finally, we verify associativity of composition. Consider three compatible morphisms
\begin{align*}
\gamma:\Delta^0\to\Delta^r,\qquad
\gamma':\Delta^r\to\Delta^s,\qquad
\gamma'':\Delta^s\to\Delta^t,
\end{align*}
where
\begin{align*}
\gamma &=(\alpha^{1,0}, \alpha^{2, 1}, \ldots,\alpha^{r,r-1}), \quad
\gamma'=(\alpha'^{r+1,r}, \alpha'^{r+2,r+1}, \ldots,\alpha'^{s,s-1}), \\
\gamma''&=(\alpha''^{s+1,s}, \alpha''^{s+2, s+1}, \ldots,\alpha''^{t,t-1}).
\end{align*}
Then, using~\eqref{eq:gammacomp} we have 
\begin{align*} 
\gamma''\circ(\gamma'\circ\gamma) &= \gamma''\circ
(\alpha^{1,0}, \alpha^{2, 1}, \ldots,\alpha^{r,r-1},\alpha'^{r+1,r},\alpha'^{r+2,r+1},\ldots,\alpha'^{s,s-1}) \\
&= (\alpha^{1,0}, \alpha^{2, 1}, \ldots,\alpha^{r,r-1},
\alpha'^{r+1,r}, \alpha'^{r+2, r+1}, \\
 & \qquad\qquad\qquad \ldots,\alpha'^{s,s-1},
\alpha''^{s+1,s},\alpha''^{s+2,s+1},\ldots,\alpha''^{t,t-1}). \label{eq:gammacomp2}
\end{align*}
On the other hand, by concatenating~$\gamma'$ and~$\gamma''$ we find 
\begin{align*} 
(\gamma''\circ\gamma')\circ\gamma &=(\alpha'^{r+1,r},\alpha'^{r+2,r+1},\ldots,\alpha'^{s,s-1},
\alpha''^{s+1,s}, \alpha''^{s+2,s+1}, \ldots,\alpha''^{t,t-1})\circ\gamma \\
&= (\alpha^{1,0}, \alpha^{2, 1}, \ldots,\alpha^{r,r-1},
\alpha'^{r+1,r}, \alpha'^{r+2, r+1}, \\
 & \qquad\qquad\qquad \ldots,\alpha'^{s,s-1},
\alpha''^{s+1,s},\alpha''^{s+2,s+1},\ldots,\alpha''^{t,t-1}). \label{eq:gammacomp3}
\end{align*}
Hence, from~\eqref{eq:gammacomp2} and~\eqref{eq:gammacomp3} we have 
\begin{align*}
\gamma''\circ(\gamma'\circ\gamma)=
(\gamma''\circ\gamma')\circ\gamma,
\end{align*}
and composition 
of morphisms 
in $\mathsf{OCS}_{\mathrm{chain}}$ is associative.

Since $\mathsf{OCS}_{\mathrm{chain}}$ has well-defined objects, identity morphisms for every object, 
hom-sets for all pairs of objects, 
and an associative composition law satisfying the identity laws, we conclude that $\mathsf{OCS}_{\mathrm{chain}}$ is a category.
\end{proof}
\subsection*{Proof of Theorem~\ref{Thm_Control_Refine}}
\label{Proof_Thm_Control_Refine}
\begin{proof}
Let
\begin{align*}
\gamma:\quad \Delta^0 \xrightarrow{\alpha^{1,0}} \Delta^1 \xrightarrow{\alpha^{2,1}} \cdots
\xrightarrow{\alpha^{r,r-1}} \Delta^r
\end{align*}
be a chain morphism in $\mathsf{OCS}_{\mathrm{chain}}$. By Definition~\ref{Def_chain_morphism}, each
\begin{align*}
\alpha^{k,k-1}=(V^{k,k-1},\mathcal I^{k,k-1}), \qquad k\in\underline r,
\end{align*}
is an elementary morphism from $\Delta^{k-1}$ to $\Delta^k$. Hence, by
Definition~\ref{Def_Ele_stor_morph}, each $V^{k,k-1}$ is a storage function associated with the
pair $(\Delta^{k-1},\Delta^k)$ in the sense of Definition~\ref{Def: Storage_Function}.

For any $t\ge 0$, the triangle inequality yields
\begin{equation}\label{Eq:proof_triangle}
\|h^r(x^r(t))-h^0(x^0(t))\|
\le
\sum_{k=1}^r \|h^k(x^k(t))-h^{k-1}(x^{k-1}(t))\|.
\end{equation}

Next, for each~$k \in \underline{r}$, since $V^{k,k-1}$ is a storage function, Definition~\ref{Def: Storage_Function} implies
that there exists a function $\lambda^{k,k-1}\in\mathcal K_\infty$ such that
\begin{align}
\label{Eq:proof_output_domination}
\lambda^{k,k-1}\bigl(\|h^k(x^k(t))-h^{k-1}(x^{k-1}(t))\|\bigr)
\le
V^{k,k-1}(x^k(t),x^{k-1}(t))
\end{align}
for all $t\ge 0$. Since $\lambda^{k,k-1}\in\mathcal K_\infty$, it is strictly increasing.
Therefore, it is invertible on $\mathbb R_{\ge 0}$, and its inverse is also strictly increasing. 
Applying $(\lambda^{k,k-1})^{-1}$ to both sides of 
\eqref{Eq:proof_output_domination} gives
\begin{align}
\label{Eq:proof_inverse_bound}
\|h^k(x^k(t))-h^{k-1}(x^{k-1}(t))\|
\le
(\lambda^{k,k-1})^{-1}\bigl(V^{k,k-1}(x^k(t),x^{k-1}(t))\bigr).
\end{align}

Substituting \eqref{Eq:proof_inverse_bound} into \eqref{Eq:proof_triangle} and summing over
$k\in\underline r$, we obtain
\begin{align*}
\|h^r(x^r(t))-h^0(x^0(t))\|
\le
\sum_{k=1}^r
(\lambda^{k,k-1})^{-1}\bigl(V^{k,k-1}(x^k(t),x^{k-1}(t))\bigr),
\end{align*}
which is exactly \eqref{eq:endpoint_bound}.
\end{proof}

\subsection*{Proof of Corollary~\ref{Cor_Control_Bounds}}
\begin{proof}
Consider any fixed $k\in\underline r$. Since $\alpha^{k,k-1}$ is an elementary morphism, the corresponding trajectories satisfy the
dissipation inequality in~\eqref{Eq:Simu_Grad} in Definition~\ref{Def: Storage_Function}. Therefore,
\begin{multline}
\label{Eq:proof_dissipation}
\dot V^{k,k-1}(t)
\leq
-\eta^{k,k-1}(V^{k,k-1}(t))
\\ +
\rho^{k,k-1}(\|u^{k-1}(t)\|) 
+ (\zeta^{k,k-1}(t))^\top X^{k,k-1}\zeta^{k,k-1}(t).
\end{multline}

Using the hypothesis that 
\begin{align}
\eta^{k,k-1}(s)\geq c^{k,k-1}s\ \textnormal{for all} \ s\geq 0, 
\end{align}
we see that~\eqref{Eq:proof_dissipation} implies
\begin{align}
\label{Eq:proof_linear_dissipation}
\dot V^{k,k-1}(t)
&\leq
-c^{k,k-1}V^{k,k-1}(t) + \rho^{k,k-1}(\|u^{k-1}(t)\|)
 +
(\zeta^{k,k-1}(t))^\top X^{k,k-1}\zeta^{k,k-1}(t) \\
&= -c^{k,k-1}V^{k,k-1}(t)+\psi^{k,k-1}(t),
\end{align}
whee we have defined~$\psi^{k,k-1}(t) = \rho^{k,k-1}(\|u^{k-1}(t)\|)
 +
(\zeta^{k,k-1}(t))^\top X^{k,k-1}\zeta^{k,k-1}(t)$. 
Multiplying both sides by $e^{c^{k,k-1}t}$ gives
\begin{align*}
\frac{d}{dt}\Big(e^{c^{k,k-1}t}V^{k,k-1}(t)\Big)
\leq
e^{c^{k,k-1}t}\psi^{k,k-1}(t), 
\end{align*}
and then integrating both sides from $0$ to $t$ yields
\begin{align*}
e^{c^{k,k-1}t}V^{k,k-1}(t)-V^{k,k-1}(0)
\leq
\int_0^t e^{c^{k,k-1}\tau}\psi^{k,k-1}(\tau)\,d\tau.
\end{align*}
Therefore,
\begin{equation} \label{eq:vkbound}
V^{k,k-1}(t) 
\leq
e^{-c^{k,k-1}t}V^{k,k-1}(0) +
\int_0^t e^{-c^{k,k-1}(t-\tau)}\psi^{k,k-1}(\tau)\,d\tau.
\end{equation}
Substituting the definition of $\psi^{k,k-1}$ into~\eqref{eq:vkbound} gives~\eqref{eq:Vk_exp_bound}, and
since $k\in\underline r$ was arbitrary, the bound holds for every storage function along the chain. 
\end{proof}   

\subsubsection*{Proof of Corollary~\ref{Cor_CL_Performance_Transfer}}
\begin{proof}
For any $t\ge 0$, the triangle inequality gives
\begin{align*}
\|h^r(x^r(t)) - \xi(t)\|
&\le
\|h^r(x^r(t)) - h^0(x^0(t))\| +
\|h^0(x^0(t)) - \xi(t)\|.
\end{align*}
Since~$\|h^0(x^0(t)) - \xi(t)\|\le \varepsilon(t)$
by hypothesis, 
it follows that
\begin{align*}
\|h^r(x^r(t)) - \xi(t)\|
\le
\|h^r(x^r(t)) - h^0(x^0(t))\|
+
\varepsilon(t).
\end{align*}
Then for all $t\geq0$, applying Theorem~\ref{Thm_Control_Refine} yields
\begin{align*}
\|h^r(x^r(t)) - \xi(t)\|
\le
\varepsilon(t) + \sum_{k=1}^r (\lambda^{k,k-1})^{-1}\!\big(V^{k,k-1}(t)\big),
\end{align*}
which proves the result.
\end{proof}

\subsection*{Proof of Lemma~\ref{Lem:PhiElementary}}
\begin{proof}
We verify the two conditions in Definition~\ref{Def: Simulation_Function}.

Since $V^{j,i}$ is a storage function from $\Delta^i$ to $\Delta^j$, Definition~\ref{Def: Storage_Function}
implies that there exists $\lambda^{j,i}\in\mathcal K_\infty$ such that
\begin{align*}
\lambda^{j,i}(\|h^j(x^j)-h^i(x^i)\|)\le V^{j,i}(x^j,x^i)
\end{align*}
for all $(x^j,x^i)$. Since the projected systems $\bar{\Delta}^i$ and $\bar{\Delta}^j$
have the same external output maps as $\Delta^i$ and $\Delta^j$, 
the inequality in 
\eqref{eq:lambdajidef} is exactly the
output domination condition in Definition~\ref{Def: Simulation_Function}.

Next, let $(x^j,x^i)$ and $u^i$ be given, and choose
\begin{align*}
u^{j\ast}=\mathcal I^{j,i}(x^j,x^i,u^i).
\end{align*}
Since $\alpha^{j,i}$ is an elementary morphism in $\mathsf{OCS}_{\mathrm{chain}}$, the storage function
$V^{j,i}$ satisfies the dissipation inequality~\eqref{Eq:Simu_Grad} in Definition~\ref{Def: Storage_Function} for this
choice of $u^{j\ast}$. For the projected systems, the internal inputs are fixed to zero, and thus we set
$w^i=0$ and $w^j=0$ in~\eqref{Eq:Simu_Grad}. 

By Definition~\ref{Def:PhiObject}, we then have
\begin{equation} \label{eq:internalinputs0}
\bar f^i(x^i,u^i)=f^i(x^i,u^i,0) \quad \textnormal{ and } \quad \bar f^j(x^j,u^{j\ast})=f^j(x^j,u^{j\ast},0).
\end{equation}
From~\eqref{eq_zeta} we have 
\begin{align*}
\zeta^{j,i}(x^j,x^i,0,0)=
\begin{bmatrix}
0\\[1mm]
y^j(x^j)-H^{j,i}y^i(x^i)
\end{bmatrix}.
\end{align*}
Hence, computing the last term in~\eqref{Eq:Simu_Grad} gives 
\begin{multline} \label{eq:newzetaij}
\zeta^{j,i}(x^j, x^i, 0, 0)^\top X^{j,i}\zeta^{j,i}(x^j, x^i, 0, 0) \\
=
\big(y^j(x^j)-H^{j,i}y^i(x^i)\big)^\top [X^{j,i}]_{22}
\big(y^j(x^j)-H^{j,i}y^i(x^i)\big)\le 0,
\end{multline}
where the inequality follows from $[X^{j,i}]_{22}\preceq 0$ in
Definition~\ref{Def: Storage_Function}.

Substituting~\eqref{eq:internalinputs0} and~\eqref{eq:newzetaij} 
into the storage function dissipation inequality in~\eqref{Eq:Simu_Grad} yields
\begin{align*}
\nabla V^{j,i}(x^j,x^i)^\top
\begin{bmatrix}
\bar f^j(x^j,u^{j\ast})\\
\bar f^i(x^i,u^i)
\end{bmatrix}
\le
-\eta^{j,i}(V^{j,i}(x^j,x^i))+\rho^{j,i}(\|u^i\|),
\end{align*}
which is exactly the dissipation inequality in Definition~\ref{Def: Simulation_Function}.
Therefore, $V^{j,i}$ is a simulation function from $\bar{\Delta}^i$ to $\bar{\Delta}^j$.

Since the same interface $\mathcal I^{j,i}$ selects the input $u^{j\ast}$ for the projected systems,
it follows that
\begin{align*}
\bar{\alpha}^{j,i}:=(V^{j,i},\mathcal I^{j,i}):\bar{\Delta}^i\to\bar{\Delta}^j
\end{align*}
is an elementary simulation morphism in $\mathsf{CS}_{\mathrm{chain}}$.
\end{proof}

\subsection*{Proof of Theorem~\ref{Thm:PhiFunctor}}
\begin{proof}
Based on~\ref{remark_free_path}, $\mathsf{OCS}_{\mathrm{chain}}$ is the free category generated by the directed multigraph whose vertices are the open control systems in $\mathsf{Sys}$ and whose edges are the elementary morphisms in Definition~\ref{Def_Ele_stor_morph}. Definition~\ref{Def:PhiObject} maps each vertex $\Delta\in\mathsf{Sys}$ to the object $\Phi(\Delta)\in\mathsf{Sys}_O$. 

According to Lemma~\ref{Lem:PhiElementary}, each elementary morphism 
\begin{align*}
\alpha^{j,i}:\Delta^i\rightarrow \Delta^j
\end{align*}
is mapped by Definition~\ref{Def:PhiMorphism} to an elementary simulation morphism
\begin{align*}
\bar{\alpha}^{j,i}:\Phi(\Delta^i)\to\Phi(\Delta^j) 
\end{align*}
in $\mathsf{CS}_{\mathrm{chain}}$. Therefore, the object map and the map on elementary morphisms define a graph morphism from the directed multigraph generating $\mathsf{OCS}_{\mathrm{chain}}$ to the underlying directed graph of $\mathsf{CS}_{\mathrm{chain}}$.

By the universal property of the free category, the graph morphism extends uniquely to a functor 
\begin{align*}
    \Phi: \mathsf{OCS}_{\mathrm{chain}} \rightarrow \mathsf{CS}_{\mathrm{chain}}.
\end{align*}
For a  chain morphism 
\begin{align*}
 \gamma=(\alpha^{1,0},\dots,\alpha^{r,r-1}),
\end{align*}
this extension is 
\begin{align*}
    \Phi(\gamma)=(\Phi(\alpha^{1,0}),\dots,\Phi(\alpha^{r,r-1})),
\end{align*}
 as given in Definition~\ref{Def:PhiMorphism}. 
 Hence, $\Phi$ preserves identity morphisms and composition. Thus, we have that 
 \begin{align*}
\Phi(\mathrm{id}_{\Delta})=\mathrm{id}_{\Phi(\Delta)},     
 \end{align*}
and 
\begin{align*}
\Phi(\gamma'\circ\gamma)=\Phi(\gamma')\circ\Phi(\gamma),
\end{align*}
for all composable chain morphisms $\gamma$ and $\gamma'$.
\end{proof}

\subsection*{Proof of Corollary~\ref{Cor:PhiTransfer_1}}
\begin{proof}
By Lemma~\ref{Lem:PhiElementary}, each elementary morphism 
of the form~$\alpha^{k+1, k} : \Delta^kj \to \Delta^{k+1}$ 
in the chain $\gamma$ induces an
elementary simulation morphism 
$\bar{\alpha}^{k+1, k} : \bar{\Delta}^k \to \bar{\Delta}^{k+1}$ 
between the corresponding projected systems without internal ports. By
Theorem~\ref{Thm:PhiFunctor}, the image $\Phi(\gamma)$ is therefore a chain morphism in
$\mathsf{CS}_{\mathrm{chain}}$. Since the proofs of Theorem~\ref{Thm_Control_Refine} and Corollary~\ref{Cor_CL_Performance_Transfer}
use only the stepwise output domination and dissipation inequalities along a chain, the same arguments 
used in their proofs also 
apply to
$\Phi(\gamma)$. Hence the same controller refinement and performance transfer bounds hold for the projected systems.
\end{proof}
\subsection*{Proof of Proposition~\ref{Prop:Composite_Control}}
\final{
\begin{proof}
By Proposition~1, the function $\boldsymbol V^{j,i}$ is a simulation function
from $\bar{\boldsymbol\Delta}^i$ to
$\bar{\boldsymbol\Delta}^j$. By hypothesis,
$\boldsymbol{\mathcal I}^{j,i}$ is an associated interface satisfying
Definition~\ref{Def_Ele_Sim_Morph}. 
Therefore, by Theorem~\ref{Def:CTRL-O-path}
\begin{align*}
\bar{\boldsymbol\alpha}^{j,i}
=
\left(
\boldsymbol V^{j,i},
\boldsymbol{\mathcal I}^{j,i}
\right)
\end{align*}
is an elementary simulation morphism in $\mathsf{CS}_{\mathrm{chain}}$.
\end{proof}}

\subsection*{Proof of Corollary~\ref{Cor:NetworkHierarchyRefinement}}
\begin{proof}
Fix $k \in \underline{r} \backslash \{1\}$. 
The subsystem-level
storage functions between
$\{\Delta_o^{k-1}\}_{o=1}^L$
and
$\{\Delta_o^k\}_{o=1}^L$
can be combined, under the interconnection conditions in \final{Proposition~\ref{Thm:CompositionalSimulation}}, into a network-level simulation function
\begin{align*}
\boldsymbol V^{k,k-1}:
\mathbb R^{\boldsymbol \sum_{o\in\underline L} n_o^k}\times
\mathbb R^{\boldsymbol \sum_{o\in\underline L} n_o^{k-1}}
\longrightarrow
\mathbb R_{\ge0}.
\end{align*}
By Proposition~\ref{Prop:Composite_Control}, this simulation function, together with its
associated composite interface
$\boldsymbol{\mathcal I}^{k,k-1}$, defines an elementary simulation morphism
\begin{align*}
\bar{\boldsymbol{\alpha}}^{k,k-1}
=
\bigl(
\boldsymbol V^{k,k-1},
\boldsymbol{\mathcal I}^{k,k-1}
\bigr):
\bar{\boldsymbol{\Delta}}^{k-1}
\longrightarrow
\bar{\boldsymbol{\Delta}}^k
\end{align*}
in $\mathsf{CS}_{\mathrm{chain}}$.

Repeating this construction for $k=2,\ldots,r$ produces a sequence of composable elementary
simulation morphisms. 
Their concatenation is therefore a chain morphism
$\bar{\boldsymbol{\gamma}}$ in $\mathsf{CS}_{\mathrm{chain}}$. The controller refinement and
performance transfer conclusions then follow directly from
Theorem~\ref{Thm_Control_Refine} and
Corollary~\ref{Cor_CL_Performance_Transfer}.
\end{proof}
\bibliographystyle{siamplain}
\bibliography{references}
\end{document}

%% file: ex_shared.tex
\usepackage{lipsum}
\usepackage{amsfonts}
\usepackage{graphicx}
\usepackage{epstopdf}
\usepackage{algorithmic}
\ifpdf
  \DeclareGraphicsExtensions{.eps,.pdf,.png,.jpg}
\else
  \DeclareGraphicsExtensions{.eps}
\fi

\newsiamremark{remark}{Remark}
\newsiamremark{example}{Example}
\newsiamremark{problem}{Problem}
\newsiamremark{hypothesis}{Hypothesis}
\crefname{hypothesis}{Hypothesis}{Hypotheses}
\newsiamthm{claim}{Claim}
\newsiamremark{fact}{Fact}
\crefname{fact}{Fact}{Facts}

\newcommand{\final}[1]{\textcolor{black}{#1}}
\newcommand{\bshe}[1]{\textcolor{black}{#1}}

\headers{Category-Theoretic Hierarchical Control}{Baike She, Tyler Hanks, James
Fairbanks, and Matthew Hale}

\title{Control Systems with Model Hierarchies: A
Category-Theoretic Framework\thanks{Submitted to the editors DATE.
\funding{This work was funded by the Fog Research Institute under contract no.~FRI-454.}}}

\author{Baike She\thanks{Georgia Institute of Technology, Atlanta, GA 30332 USA  
  (\email{bshe6@gatech.edu}).} 
\and Tyler Hanks\thanks{Florida Institute of Technology, Melbourne, FL 32904 USA 
  (\email{thanks@fit.edu}).}
\and James Fairbanks\thanks{University of Florida, Gainesville, FL 32611 USA 
  (\email{fairbanksj@ufl.edu}).}
\and Matthew Hale\thanks{Georgia Institute of Technology, Atlanta, GA 30332 USA  
  (\email{mhale30@gatech.edu}).}
}

\usepackage{amsopn}
